\documentclass[10pt,reqno]{amsart}

\usepackage{amssymb,amsmath,amsthm,amsxtra,amsfonts}
\usepackage{graphicx}

\usepackage{hyperref}

\usepackage{setspace}

\newtheorem{theorem}{Theorem}[section]

\newtheorem{lemma}[theorem]{Lemma}

\numberwithin{equation}{section}
\newtheorem{definition}[theorem]{Definition}

\theoremstyle{remark}
\newtheorem*{remarks}{Remarks}
\newtheorem*{remark}{Remark}

\newcommand{\bbC}{{\mathbb{C}}}
\newcommand{\bbD}{{\mathbb{D}}}

\newcommand{\bbR}{{\mathbb{R}}}

\newcommand{\bbZ}{{\mathbb{Z}}}
\newcommand{\uc}{[\texttt{OPUC}]}
\newcommand{\rl}{[\texttt{OPRL}]}

\newcommand{\fre}{{\frak{e}}}
\newcommand{\frf}{{\frak{f}}}

\newcommand{\calC}{{\mathcal{C}}}

\newcommand{\calJ}{{\mathcal J}}

\newcommand{\calM}{{\mathcal M}}

\newcommand{\calS}{{\mathcal S}}
\newcommand{\calT}{{\mathcal T}}

\newcommand{\til}{\tilde  }

\newcommand{\Arg}{\text{\rm{Arg}}}

\newcommand{\ran}{\text{\rm{Ran}\,}}
\newcommand{\sgn}{\text{\rm{sgn}}}

\newcommand{\beq}{\begin{equation}}
\newcommand{\eeq}{\end{equation}}
\newcommand{\ba}{\begin{align*}}
\newcommand{\ea}{\end{align*}}
\newcommand{\veps}{\varepsilon}
\newcommand{\Deg}{\operatorname{Deg}}

\DeclareMathOperator{\real}{Re}
\DeclareMathOperator{\imag}{Im}

\DeclareMathOperator*{\esssup}{ess\,supp}

\DeclareMathOperator*{\res}{Res}
\DeclareMathOperator{\Int}{Int}
\DeclareMathOperator{\sg}{sg}

\begin{document}

\title{Finite Range Perturbations of \\ Finite Gap Jacobi and CMV Operators}
\author{Rostyslav Kozhan}
\thanks{The project was supported by the grant KAW 2010.0063 from the Knut and Alice Wallenberg Foundation}
\email{kozhan@math.uu.se}
\address{Uppsala University; Department of Mathematics; Box 480, 75106 Uppsala, Sweden}


\begin{abstract}
Necessary and sufficient conditions are presented for a measure to be the spectral measure of a finite range perturbation of a Jacobi or CMV operator from a finite gap isospectral torus.
The special case of eventually periodic operators solves an open problem of Simon~\cite[D.2.7]{OPUC2}.

We also solve the inverse resonance problem: it is shown that an operator is completely determined by the set of its eigenvalues and resonances, and we provide necessary and sufficient conditions on their configuration for such an operator to exist.
\end{abstract}

\maketitle






\section{Introduction}\label{sIntro}
By a Jacobi operator/matrix we will call a bounded Hermitian operator on $\ell_2(\bbZ_+)$ of the form
\begin{equation}\label{jacobi}
\calJ=\left(
\begin{array}{cccc}
b_1&a_1&{0}&\\
a_1 &b_2&a_2&\ddots\\
{0}&a_2 &b_3&\ddots\\
&\ddots&\ddots&\ddots\end{array}\right).
\end{equation}
Any operator of the form \eqref{jacobi} will be denoted by $\calJ[a_n,b_n]_{n=1}^\infty$. Sequences $\{a_n\}$, $\{b_n\}$ are called the Jacobi parameters of $\calJ$. We always assume these are bounded sequences, and $a_n>0$, $b_n\in\bbR$ for all $n$.

Associated to $\calJ$, we have $\mu$, the spectral measure of $\calJ$ with respect to the vector $e_1:=(1,0,0,\ldots)^T$ (which is cyclic since all $a_j>0$):
\begin{equation}\label{spectralDefinition}
\int_\bbR f(x)d\mu(x) = \langle e_1, f(\calJ) e_1 \rangle.
\end{equation}
Conversely, given any probability measure $\mu$ with compact and not finite support in $\bbR$, we can form the sequence of orthonormal polynomials which satisfy the three-term recurrence relation with the coefficients $\{a_n,b_n\}_{n=1}^\infty$ from~\eqref{jacobi}.

In this paper we will consider only measures with  essential support equal to a finite gap set
\begin{equation}\label{e}
 \fre = \bigcup_{j=1}^{l+1} [\alpha_j,\beta_j], \quad \alpha_1<\beta_1 <\alpha_2<\ldots<\alpha_{l+1}<\beta_{l+1}.
\end{equation}
We will refer to each $[\alpha_j,\beta_j]$ as a ``band'', and to each $[\beta_j,\alpha_{j+1}]$ as a ``gap''. $l$ here is the number of gaps.

Associated to $\fre$ is a natural class of operators called the isospectral torus $\calT_\fre$ of Jacobi operators (defined in Definition~\ref{defTorus} below). This includes as special cases the free Jacobi operator (discrete Schr\"{o}dinger operator) when $l=0$, $\fre=[-2,2]$ and, more generally, all periodic Jacobi operators when harmonic measures of each $[\alpha_j,\beta_j]$ in $\fre$ are rational. If not all of these harmonic measures are rational, then $\calT_\fre$ consists of almost-periodic Jacobi operators (see more details in Subsection~\ref{ssPeriodic}).

Operators in the isospectral torus are well-studied by now, and we propose to go one step further and consider their finite range perturbations: take $\calJ\in\calT_\fre$ and change finitely many of its Jacobi coefficients. 

\bigskip

Similar construction is also considered for measures on the unit circle. By a CMV operator/matrix we will call a unitary operator on $\ell_2(\bbZ_+)$ of the form
$$
\calC=
\left(
\begin{array}{cccccc}
\bar{\alpha}_0 & \rho_0   & 0   &  0 & 0 & \vphantom{\ddots}\\
\rho_0        & -\alpha_0 &  0  & 0 &  0 & \vphantom{\ddots}\\
0            &  0        &  \bar{\alpha}_2 & \rho_2  & 0 & \vphantom{\ddots} \\
0            &  0        &  \rho_2        & -\alpha_2  & 0 & \vphantom{\ddots} \\
0            &  0        &   0           &   0         & \bar{\alpha}_4 & \vphantom{\ddots} \\
\hphantom{\ddots} &     \hphantom{\ddots}        & \hphantom{\ddots}          &    \hphantom{\ddots}          &     \hphantom{\ddots}      & \ddots
\end{array}
\right)
\left(
\begin{array}{cccccc}
1  & 0 & 0 & 0 & 0 & \vphantom{\ddots} \\
0 & \bar{\alpha}_1 & \rho_1   & 0   &  0 &  \vphantom{\ddots} \\
0 & \rho_1        & -\alpha_1 &  0  & 0 &    \vphantom{\ddots}\\
0 & 0            &  0        &  \bar{\alpha}_3 & \rho_3  &  \vphantom{\ddots}\\
0 & 0            &  0        &  \rho_3        & -\alpha_3  &  \vphantom{\ddots} \\
\hphantom{\ddots} &     \hphantom{\ddots}        & \hphantom{\ddots}          &    \hphantom{\ddots}          &     \hphantom{\ddots}      & \ddots
\end{array}
\right),
$$
where $\rho_n := \sqrt{1-|\alpha_n|^2}$.   An operator of this form will be denoted by $\calC[\alpha_n]_{n=0}^\infty$. Coefficients $\alpha_n$ are assumed to satisfy $|\alpha_n|<1$ and are called the Verblunsky coefficients. The spectral measure of $\calC$ with respect to $e_1$ now lives on the unit circle $\partial\bbD:=\{z\in\bbC:|z|=1\}$:
\begin{equation}\label{spectralDefinition2}
\int_{0}^{2\pi} f(e^{i\theta})d\mu(\theta) = \langle e_1, f(\calC) e_1 \rangle.
\end{equation}
Conversely, given any probability measure $\mu$ on the unit circle not supported on finitely many points, we can form a sequence of orthogonal polynomials that satisfy Szeg\H{o}'s recurrence relations which allow to recover the Verblunsky coefficients.

A finite gap set on the unit circle is defined by
\begin{equation}\label{e2}
 \frf = \{e^{i\theta}: \theta\in \cup_{j=1}^{l} [\theta_{2j-1},\theta_{2j}] \}, \quad \theta_1<\theta_2 <\theta_3<\ldots<\theta_{2l-1}<\theta_{2l}<\theta_{1}+2\pi.
\end{equation}
We will refer to each $[\theta_{2j-1},\theta_{2j}]$ (as well as to its image under $\theta\mapsto e^{i\theta}$) as a ``band'', and to the intervals between them as ``gaps''. The number of gaps is $l$.

The associated isospectral torus $\calT_\frf$ of CMV operators is defined in Definition~\ref{defTorus} below. This includes the free CMV operator (bilateral shift on $\ell_2(\bbZ_+)$) when $\frf=\partial\bbD$ (one should think of it as $l=0$ in~\eqref{e2}), 
as well as all periodic CMV operators and certain almost-periodic operators (see Subsection~\ref{ssPeriodic}).


We will study here finite range perturbations of these operators: take an operator from $\calC\in\calT_\frf$ and change finitely many of its Verblunsky coefficients. 

\medskip

The main result of the current paper is the if-and-only-if criterion for the spectral measures, Theorem~\ref{thmS}. It is remarkable that this classification was not previously known
even for the simplest case of finite range perturbations of the free Jacobi operator (another, less direct, proof is delegated to the author's separate manuscript~\cite{K_spectral}), despite the fact that it is by far the most well studied Jacobi operator.

Moreover, we provide the if-and-only-if description of the finite range perturbations from four points of views: from the point of view of operators (finite range perturbations), from the point of view of spectral measures, from the point of view of eigenvalues and resonances (``Dirichlet data''), and from the point of view of meromorphic functions on Riemann surfaces ($m$-functions and Carath\'{e}odory functions).

In particular the classification of Carath\'{e}odory functions solves an open question from Simon~\cite[D.2.7, p.981]{OPUC2}, and the classification of eigenvalues and resonances solves existence and uniqueness of the inverse resonance problem.


\medskip

The organization of the paper is as follows. We review some previously known results in Subsection~\ref{ssHistory}. We continue with a rather lengthy introduction that includes all the definitions and preliminaries in Section~\ref{sPrelim}. In Section~\ref{sMandF} we classify the $m$-functions and Carath\'{e}odory functions of our operators. In Section~\ref{sSpectral} we deduce the spectral theorem. In Section~\ref{sResonance} we show existence and uniqueness of the resonance problem. In the final Section~\ref{sInterpolation} we provide an explicit description of $m$-functions in terms of its poles as a solution to an interpolation type problem.

\medskip

The theories of orthogonal polynomials on the real line (OPRL) and on the unit circle (OPUC) are closely related. 
We will be discussing the results for Jacobi and CMV operators in parallel, labeling each of the results  with $\rl$ and $\uc$, respectively. One of the joys of writing this paper was in appreciating the similarities between these two theories, while at the same time dealing with the subtle differences between them. We hope the reader finds this enjoyable too.

\smallskip

\textbf{Acknowledgements.} The work was finished during the author's stay at the Royal Institute of Technology (KTH). The author would like to thank the Department of Mathematics, and especially Kurt Johansson, for the hospitality. It is also a pleasure to thank Rowan Killip (UCLA) for his insightful comments. 

\subsection{History}\label{ssHistory}

Finite gap Jacobi and CMV operators appear in connection with the polynomials orthogonal with respect to a measure supported on a system of curves in $\bbC$. We refer the reader to the papers by Widom~\cite{Wid69}, Aptekarev~\cite{Apt84}, Sodin--Yuditskii~\cite{SodYud97}, Peherstorfer--Yuditskii~\cite{PehYud03}, Christiansen--Simon--Zinchenko~\cite{CSZ1,CSZ2,CSZ3,CSZr} and references therein. 

Spectral measures for short-range perturbations of the free Jacobi operator were studied by numerous authors, among which we would like to distinguish the results of Geronimo--Case~\cite{GC}, Geronimo~\cite{Geronimo}, and Damanik--Simon~\cite{DS2}. Spectral properties of finite range perturbations of periodic Jacobi operators were the subject of Geronimo--Van Assche~\cite{GvA86} and Iantchenko--Korotyaev~\cite{IK3}.

The explicit if-and-only-if characterization of the spectral measures for finite range perturbations of the free and periodic Jacobi operators (with all gaps open) is shown in the author's manuscript~\cite{K_spectral}.
In that paper we are able to classify spectral measures not only of finite range perturbations, but also of super-exponential and exponential ones. The current paper contains a much simpler and straightforward proof for the finite range case, that does not require the author's lengthy route~\cite{K_jost,K_merom,K_spectral}  through the matrix-valued spectral problem via the Damanik--Killip--Simon~\cite{DKS} ``Magic'' formula. But the real strength of the current approach is that it lends itself to the perturbations of operators from the isospectral torus not only of periodic operators but for any finite gap set. Moreover, the unitary analogue can be proven in the same way with only slight variations, in particular solving the open problem~\cite[D.2.7, p.981]{OPUC2}.

The direct resonance problem for finite range perturbations of periodic Jacobi operators was completely solved in Iantchenko--Korotyaev~\cite[Thm 1.2]{IK3}\footnote{\cite[Thm 1.2]{IK3} has a mistake: part (2) should not be there}. Their inverse resonance problem assumed additional information. Uniqueness for the inverse resonance problem for super-exponential perturbations of the free Jacobi operator was solved by Brown--Naboko--Weikard~\cite{BNW_Jac}. Existence and uniqueness for the inverse resonance problem for the super-exponential perturbations of the free and periodic Jacobi operators is solved by the author in~\cite{K_spectral}. We would also like to mention the results by Marletta--Weikard~\cite{MW}, Marletta--Naboko--Shterenberg--Weikard~\cite{MNSW}, and the author~\cite{K_spectral}, that study the stability of this inverse resonance problem.

Let us review the results for the OPUC case now. The spectral measures for finite range perturbations of the free CMV operator were fully understood for quite awhile now: these have the name of the Bernstein--Szeg\H{o} measures, and in the current context they seem to have first appeared in the papers by Verblunsky~\cite{Ver36} and then later Geronimus~\cite{Geronimus44,bGeronimus}. Finite range perturbations of periodic (and ``periodic up to a phase'', see Subsection~\ref{ssPeriodic} below) CMV operators were studied by Peherstorfer--Steinbauer~\cite{PehSte86b}.

The uniqueness for the inverse resonance problem for the super-exponential perturbations of the free CMV operator was established by Weikard--Zinchenko~\cite{WZ}. Stability for this problem was obtained by Shterenberg--Weikard--Zinchenko~\cite{SWZ}.

\smallskip

For a textbook presentation and a more extensive history overview for the theory of orthogonal polynomials on the real line (including the spectral theory of periodic and finite gap Jacobi operators), we refer the reader to the recent Simon's monograph~\cite{Rice}.  For the theory of orthogonal polynomials on the unit circle, 
we refer to~\cite{OPUC1,OPUC2}. We follow closely the terminology there.

\section{Preliminaries}\label{sPrelim}

Let us assume for the rest of the paper that  $l\ne 0$ for the \uc{} case  (unless specified otherwise). The case $l=0$ (that is, $\frf=\partial\bbD$) can be easily accommodated, but since it is easy and solved (Bernstein--Szeg\H{o}), let us ignore it, so that we can assume that the Riemann surface $\calS_\frf$, see Def.~\ref{defS}, is connected.

\subsection{Two-sheeted Riemann surfaces}\label{ssRiem}

Let $\bbC_+ = \{z: \imag z>0\}$, $\bbC_- = \{z:\imag z<0\}$, $\bbD=\{z:|z|<1\}$.

\begin{definition}\label{defS}
\hspace*{\fill}\\
\indent {\normalfont $\rl$} Assume $\fre$ is a finite gap set~\eqref{e}. Define $\calS_\fre$ to be the Riemann surface obtained by gluing two copies, $\calS_{\fre,+}$ and $\calS_{\fre,-}$, of $\bbC\cup\{\infty\}$ with a slit along $\fre$ $($include $\fre$ as a top edge and exclude it from the lower$)$  in the following way: passing from $\calS_{\fre,+}\cap\bbC_+ $ through $\fre$ takes us to
$ \calS_{\fre,-} \cap\bbC_- $, and from $ \calS_{\fre,+}\cap\bbC_-$ to $ \calS_{\fre,-}\cap\bbC_+$.

\smallskip{\normalfont $\uc$} Assume $\frf$ is a finite gap set~\eqref{e2}. Define $\calS_\frf$ to be the Riemann surface obtained by gluing two copies, $\calS_{\frf,+}$ and $\calS_{\frf,-}$, of $\bbC\cup\{\infty\}$ with a slit along $\frf$ $($include $\frf$ as an edge of $\bbD$ and exclude it from the edge of $\bbC\setminus\bbD)$ in the following way: passing from $\calS_{\frf,+}\cap\bbD $ through $\frf$ takes us to
$ \calS_{\frf,-} \cap\bbC\setminus\bbD$, and from $ \calS_{\frf,+}\cap\bbC\setminus\bbD$ to $ \calS_{\frf,-}\cap\bbD$.

\end{definition}
\begin{remark}
 $\calS_\fre$ is topologically a sphere with $l$ handles, while $\calS_\frf$ is topologically a sphere with $l-1$ handles.

\end{remark}

Let $\pi:\calS_\fre\to \bbC\cup\{\infty\}$ be the ``projection map'' which extends the natural inclusions $\calS_{\fre,+}\hookrightarrow \bbC\cup\{\infty\}$, $\calS_{\fre,-}\hookrightarrow \bbC\cup\{\infty\}$. We will also use $\pi$ to denote the analogous projection map $\calS_\frf\to\bbC\cup\{\infty\}$.

\begin{definition}\label{defSharp}
\hspace*{\fill}\\
\indent {\normalfont $\rl$}
\hspace*{\fill}
\begin{itemize}
\item[$\bullet$] For $z\in\bbC\cup\{\infty\}$, denote by $z_+$ and $z_-$ the two preimages $\pi^{-1}(z)$ in $\calS_{\fre,+}$ and $\calS_{\fre,-}$ respectively $($for $z\in\cup_{j=1}^{l+1}\{\alpha_j,\beta_j\}$, $z_+$ and $z_-$ coincide$)$.
\item[$\bullet$] Let $\tau:\calS_{\fre}\to\calS_\fre$ be the map that maps $z_+$ to $z_-$ and $z_-$ to $z_+$ for all $z\in\bbC\cup\{\infty\}$.
\item[$\bullet$] For a function $m$ on $\calS_\fre$, let  $m^\sharp(z)=m(\tau(z))$.
\end{itemize}

\smallskip{\normalfont $\uc$}
\begin{itemize}
\item[$\bullet$] For $z\in\bbC\cup\{\infty\}$, denote by $z_+$ and $z_-$ the two preimages $\pi^{-1}(z)$ in $\calS_{\frf,+}$ and $\calS_{\frf,-}$ respectively $($for $z\in\cup_{j=1}^{2l}\{e^{i\theta_j}\}$, $z_+$ and $z_-$ coincide$)$.
\item[$\bullet$] Let $\tau:\calS_{\frf}\to\calS_\frf$ be the map that maps $z_+$ to $z_-$ and $z_-$ to $z_+$ for all $z\in\bbC\cup\{\infty\}$.
\item[$\bullet$] For a function $F$ on $\calS_\frf$, let  $F^\sharp(z)=F(\tau(z))$.
\end{itemize}
\end{definition}

\subsection{Meromorphic functions on $\calS$}\label{ssMeromorphic}

\hspace*{\fill}\\
\indent {\normalfont $\rl$}
$\calS_\fre$ is a Riemann surface and has an associated notion of analyticity for functions $f:\calS_\fre \to \bbC$. For points $z_0\in\pi^{-1}(\bbC\setminus\cup_{j=1}^{l+1}\{\alpha_j,\beta_j\})$ we can always find a neighborhood $U$ of $z_0$ in $\calS_\fre$ on which the projection $\pi$ is one-to-one onto $\pi(U)$. Analyticity of $f$ at $z_0$ becomes equivalent to analyticity of $f(\pi^{-1}(z)):\bbC\to\bbC$ at $\pi(z_0)$. For an endpoint $z_0\in\pi^{-1}(\cup_{j=1}^{l+1}\{\alpha_j,\beta_j\})$, a function is analytic at $z_0$ if in a small neighborhood of $z_0$ on $\calS_\fre$ it can be expanded into Taylor's series
$$
f(z) = \sum_{j=0}^\infty k_j (z-z_0)^{j/2},
$$
where one fixes any branch of the square root for $z\in\calS_{\fre,+}$ and its negative for $z\in\calS_{\fre,-}$. Similarly one defines the notion of meromorphic functions.

Let us take the polynomial
\begin{equation}\label{rfre}
R_\fre(z):=\prod_{j=1}^{l+1}(z-\alpha_j)(z-\beta_j),
\end{equation}
and choose the branch of $\sqrt{R_\fre(z)}$, analytic on $\bbC\setminus\fre$, that is positive on $(\beta_{l+1},+\infty)$. Now define $\sqrt{\tilde{R}_\fre(z)}$ to be the function $\calS_\fre \to \bbC\cup\{\infty\}$ equal to $\sqrt{R_\fre(z)}$ on $\calS_{\fre,+}$ and to $-\sqrt{R_\fre(z)}$ on $\calS_{\fre,-}$. Easy to see then that this function is analytic on $\pi^{-1}(\bbC)$ and meromorphic on $\calS_\fre$. We will start using the same symbol $\sqrt{R_\fre}$ instead of $\sqrt{\til{R}_\fre}$ and hope this will not cause a confusion.

For a future reference we note that $\sqrt{R_\fre(x_+)}$ belongs to $(-1)^{l+1-k} \bbR_+$ for  $x\in(\beta_k,\alpha_{k+1})$ and to $(-1)^{l+1-k} i \bbR_+$ for $x\in(\alpha_k,\beta_k)$.

It is not hard to check (see~\cite[Prop 5.12.1]{Rice}) that any function that is meromorphic on the whole surface $\calS_\fre$ is of the form
\begin{equation}\label{pqr}
g(z)=\frac{p(z)+q(z)\sqrt{R(z)}}{a(z)}
\end{equation}
for some polynomials $p,q,a$ ($a \not\equiv 0$) that have no common zeros.

In the last formula and everywhere further in the text, whenever $z\in \calS_\fre$ and $p:\bbC\to\bbC$ is a function of a complex variable, we will routinely write $p(z)$ instead of the actual $p(\pi(z))$.

Note that if $g$ is~\eqref{pqr}, then $g^\sharp(z)$ is given by the same expression by with the minus sign in front of $\sqrt{R(z)}$.

For any function $g$ meromorphic on $\calS_\fre$ and any $a\in\bbC\cup\{\infty\}$, the number of solutions of $g(z)=a$ is independent of the $a$, if we count the solutions with multiplicities. We call this common integer the degree of $g$ and denote it by $\deg g$. We will use $\Deg p$ to denote the conventional notion of degree of a polynomial $p$. We stress that multiplicities at a branch point $z_0\in\pi^{-1}(\cup_{j=1}^{l+1}\{\alpha_j,\beta_j\})$ should be counted in powers of $(z-z_0)^{1/2}$, not $(z-z_0)$. E.g., if $g(z) = a+ (z-z_0)^{j/2} h(z)$ with $h(z_0)\ne 0$, $z_0\in\pi^{-1}(\cup_{j=1}^{l+1}\{\alpha_j,\beta_j\})$, then $z_0$ is the solution of $g(z)=a$ of multiplicity $j$, not $j/2$.

\smallskip

\uc{} The notion of analyticity/meromorphicity works in the same way  for $\calS_\frf$ as for $\calS_\fre$. The analogue of~\eqref{rfre} is the polynomial
$$
R_\frf(z):=\pm\prod_{j=1}^{2l} e^{-i\theta_j/2} (z-e^{i\theta_{j}}),
$$
where the sign is chosen so that $\{z\in\partial\bbD : z^{-l}R_\frf(z) \le 0 \} = \frf$. Indeed,
\begin{equation}\label{reality}
e^{-il\theta}R_\frf(e^{i\theta}) = \pm 2^{2l} (-1)^l \prod_{j=1}^{2l} \sin\tfrac{\theta-\theta_j}{2},
\end{equation}
which is real and of the same sign on $\frf$.

If $l$ is even then on $\calS_{\frf,+}$ we pick the square root in $\sqrt{R_\frf(z)}$ that satisfies $\imag z^{-l/2} \sqrt{R_\frf(z)} \ge 0$  for $z=e^{i\theta}, \theta\in [\theta_1,\theta_2]$, and we extend it to $\calS_{\frf,-}$ by defining  $\sqrt{R_\frf(z)}^\sharp = - \sqrt{R_\frf(z)}$. Such a function is analytic on $\pi^{-1}(\bbC)$ and meromorphic on $\calS_\frf$.

If $l$ is odd, we first take $z^{-1/2}$ with the branch cut $e^{i\theta_1} \bbR_+$, and  in $\sqrt{R_\frf(z)}$ we pick the branch of the square root that has $\imag z^{-l/2} \sqrt{R_\frf(z)} \ge 0$  for $z=e^{i\theta}, \theta\in [\theta_1,\theta_2]$ (alternatively, one can also use the ``sieving'' idea, see~\cite[Sect 11.7]{OPUC2}). Then we extend $\sqrt{R_\frf(z)}$ to $\calS_{\frf,-}$ by defining  $\sqrt{R_\frf(z)}^\sharp = - \sqrt{R_\frf(z)}$. Such a function is analytic on $\pi^{-1}(\bbC)$ and meromorphic on $\calS_\frf$.

For a future reference we note that $e^{-il\theta/2}\sqrt{R_\frf((e^{i\theta})_+)}$ belongs to $(-1)^{k-1} i \bbR_+$ for  $\theta\in(\theta_{2k-1},\theta_{2k})$ and to $(-1)^{k-1} \bbR_+$ for $\theta\in(\theta_{2k},\theta_{2k+1})$, $1\le k \le l$.

\subsection{Periodic and almost periodic operators}\label{ssPeriodic}
We call a Jacobi (CMV) operator  periodic if its Jacobi (Verblunsky) coefficients are periodic, that is, there exists $p\ge 1$ such that $a_{n+p} = a_n$ and $b_{n+p} = b_n$ ($\alpha_{n+p} = \alpha_n$) for all $n$. 
 For the special case of constant coefficients (that is, $p=1$) we call these the free Jacobi and the free CMV operator, respectively.

We will call a sequence $\{s_j\}_{j=1}^\infty$ quasiperiodic with at most $q$ quasiperiods if there exists a continuous function $f$ on the $q$-torus $\partial\bbD^q$ and real numbers (quasiperiods) $w_1,\ldots,w_q$ such that $s_j = f(e^{i j w_1},\ldots,e^{ij w_q})$.

Accordingly, we will refer to a Jacobi (CMV) operator as almost periodic with $q$ quasiperiods if its Jacobi (Verblunsky) coefficients are quasiperiodic with at most $q$ quasiperiods. One should think of $p$-periodic operators as that special case of almost periodic operators with at most $p$ quasiperiods when all quasiperiods are integer multiples of $\tfrac{2\pi}{p}$.

For a future reference, notice that just like for periodic operators, knowing $\calJ[a_n,b_n]_{n=N_0}^\infty$ or $\calC[\alpha_n]_{n=N_0}^\infty$ of an almost periodic operator uniquely determines the full operator $\calJ[a_n,b_n]_{n=1}^\infty$ or $\calC[\alpha_n]_{n=0}^\infty$. In fact one can uniquely extend it to the two-sided almost periodic operator on $\ell^2(\bbZ)$.

It is well known that the essential spectrum of a periodic Jacobi (CMV) operator is a finite gap set. The essential spectrum of an almost periodic could be a finite gap or an infinite gap set.

Given a finite gap set $\fre$ (or $\frf$) one may ask whether it can be the essential spectrum of a periodic or almost periodic operator. The answer is always yes, and in fact, there exists a whole $l$-dimensional set (topologically an $l$-dimensional torus $(S^1)^l$) of such operators that we will refer to as the isospectral torus. The following classifies when these operators are periodic or almost periodic:

\smallskip
\rl{}
\begin{itemize}
\item If each interval $[\alpha_j,\beta_j]$ in $\fre$ has rational harmonic measure, then there exists a periodic Jacobi operator with $\fre$ as its essential spectrum.
\item If one of the intervals $[\alpha_j,\beta_j]$ in $\fre$ has irrational harmonic measure, then there exists an almost periodic with at most $l$ quasiperiods Jacobi operator  with $\fre$ as its essential spectrum.
\end{itemize}

\uc{}
\begin{itemize}
\item If each band in $\frf$ has rational harmonic measure and $\prod_{j=1}^{2l} e^{i\theta_j} = 1$, then there exists a periodic CMV operator with $\frf$ as its essential spectrum.
\item If each band in $\frf$ has rational harmonic measure and $\prod_{j=1}^{2l} e^{i\theta_j} \ne 1$, then there exists a CMV operator, periodic up to a phase (that is, $\alpha_{n+p} = \lambda \alpha_{n}$ for some $\lambda\in\partial\bbD$), with $\frf$ as its essential spectrum.\footnote{This is of course just a special case of an almost periodic operator with $\Arg \lambda$ as a quasiperiod.}
\item If one of the bands in $\frf$ has irrational harmonic measure, then there exists an almost periodic with at most $l$ quasiperiods CMV operator with $\frf$ as its essential spectrum.
\end{itemize}

\subsection{Herglotz and Carath\'{e}odory functions}
To each Jacobi operator $\calJ$ and its spectral measure $\mu$,~\eqref{spectralDefinition}, we can associate
\begin{equation}\label{m}
m(z):=\int_\bbR \frac{d\mu(x)}{x-z}, \quad z\notin\esssup\mu,
\end{equation}
the Borel/Stieltjes/Cauchy transform of $\mu$. From~\eqref{spectralDefinition}, $m$ is also the $(1,1)$-entry of the resolvent of $\calJ$. We will refer to this function as the $m$-function of $\calJ$.

$m$ is a Herglotz function, meaning that $\imag m(z)>0$ whenever $\imag z>0$, and $\imag m(z)<0$ whenever $\imag z<0$. It follows from the definition that
\begin{equation}\label{symmM}
\overline{m(\bar{z})} = m(z).
\end{equation}

Let us introduce the notation
\begin{equation*}
\calJ^{(s)} := \calJ[a_{n+s},b_{n+s}]_{n=1}^\infty,
\end{equation*}
that is, $\calJ^{(s)}$ is the Jacobi matrix obtained from $\calJ$ by removing the first $s$ rows and columns. In particular, $\calJ^{(0)}$ is just $\calJ$.

The $m$-functions $m$ and $m^{(1)}$ of $\calJ$ and $\calJ^{(1)}$ are known to obey
\begin{equation}\label{stripping}
a_1^2 m^{(1)}(z) = b_1-z-\frac{1}{m(z)}.
\end{equation}
Indeed this follows immediately from the Schur complement formula.

\smallskip

To each CMV operator $\calC$ and its spectral measure $\mu$,~\eqref{spectralDefinition2}, we can associate
\begin{equation}\label{F}
F(z):=\int_0^{2\pi} \frac{e^{i\theta} + z}{e^{i\theta} -z}d\mu(\theta), \quad z\notin\esssup\mu,
\end{equation}
which we will call the Carath\'{e}odory function of $\mu$. From~\eqref{spectralDefinition2}, $\frac{F(z) - 1}{2z}$ is the $(1,1)$-entry of the resolvent of $\calC$.

$F$ is a Carath\'{e}odory function, by which we mean a function satisfying $\real F(z)>0$ whenever $z\in\bbD$, $\real F(z)<0$ whenever $z\in\bbC\setminus\overline{\bbD}$,  and $F(0)=1$. The counterpart of~\eqref{symmM} is
\begin{equation}\label{symmF}
\overline{F(1/\bar{z})} = - F(z),
\end{equation}
which follows immediately from~\eqref{F}.

If $\calC=\calC[\alpha_n]_{n=0}^\infty$ define $\calC^{(s)} := \calC[\alpha_{n+s}]_{n=0}^\infty$. The Carath\'{e}odory functions $F$ and $F^{(1)}$ of $\calC$ and $\calC^{(1)}$ are known (see, e.g.,~\cite[Eq. (11.7.73)]{OPUC2}) to satisfy
\begin{equation}\label{strCar}
\frac{F^{(1)}(z)+1}{F^{(1)}(z)-1} = \frac{z(F(z)+1)- \bar\alpha_0(F(z)-1)}{-z\alpha_0(F(z)+1) + (F(z)-1)}.
\end{equation}

In the next lemma we show how one can recover the absolutely continuous and pure point parts of the measure from knowing $m$ or $F$.

\begin{lemma}[Herglotz Representation Theorem]\label{herglotz}

\hspace*{\fill}\\
\indent {\normalfont $\rl$}
Let $m$ be~\eqref{m} for some probability measure $\mu$ on $\bbR$. Then the absolutely continuous part of $\mu$ can be recovered by
\begin{equation}\label{herg1}
\frac{d\mu}{dx}=\frac{1}{\pi} \lim_{\veps\downarrow0} \imag m(x+i \veps),
\end{equation}
$($Lebesgue a.e.$)$, and the pure point part by
\begin{equation}\label{herg2}
\mu(\{\lambda\})=\lim_{\veps \downarrow 0} \veps\, \imag m(\lambda+i\veps) = \tfrac{1}{i} \lim_{\veps\downarrow 0} \veps\, m(\lambda+i\veps).
\end{equation}

\smallskip{\normalfont $\uc$}
Let $F$ be~\eqref{F} for some probability measure $\mu$ on $\partial\bbD$. Then the absolutely continuous part of $\mu$ can be recovered by
\begin{equation}\label{cara1}
\frac{d\mu}{d\theta}=\frac{1}{2\pi} \lim_{r \uparrow 1} \real F(r e^{i\theta}),
\end{equation}
$($Lebesgue a.e.$)$,
and the pure point part by
\begin{equation}\label{cara2}
\mu(\{\theta\})=\lim_{r \uparrow 1} \left( \frac{1-r}{2}\right) F(r e^{i\theta}).
\end{equation}
\end{lemma}

%

We will be particularly interested in the Herglotz and Carath\'{e}odory functions that have a meromorphic continuations from $\bbC\setminus\fre$ to $\calS_\fre$ and from $\bbC\setminus\frf$ to $\calS_\frf$. The following subclass deserves a special name. These are precisely the $m$-functions and Carath\'{e}odory functions of Jacobi and CMV operators from Subsection~\ref{ssPeriodic}.
\begin{definition}\label{defMinimal}

\hspace*{\fill}\\
\indent {\normalfont $\rl$}
A minimal Herglotz function on $\calS_\fre$ is a function that is meromorphic on $\calS_\fre$ and obeys
\begin{itemize}
\item[(i)] $m$ restricted to $\calS_{\fre,+}$ satisfies~\eqref{m} for some probability measure $\mu$ on $\bbR$; 
\item[(ii)] $\deg m = l+1$;
\item[(iii)] $m$ has a pole at $\infty_-$.
\end{itemize}

\smallskip{\normalfont $\uc$}
A minimal Carath\'{e}odory function on $\calS_\frf$ is a function that is meromorphic on $\calS_\frf$ and obeys
\begin{itemize}
\item[(i)] $F$ restricted to $\calS_{\frf,+}$ satisfies~\eqref{F} for some probability measure $\mu$ on $\partial\bbD$;
\item[(ii)] $\deg F = l$.
\end{itemize}
\end{definition}
\begin{remarks}
1. For $l=0$, \uc{}, the condition (ii) should be interpreted as $F=\operatorname{const}$. 

2. We note that conditions (3) and (4) in the definition of minimal Carath\'{e}odory function in~\cite[p. 767]{OPUC2} are in fact automatic from the condition (i),~\eqref{pqr}, and~\eqref{symmF}.

3. The term ``minimal'' comes from the fact that any function on $\calS_\fre$ of the form~\eqref{pqr} with $q\not\equiv 0$ has degree $l+1$ or higher (for $\calS_\frf$ --- degree $l$ or higher).

4. The degree condition (ii) implies (see~\cite[Thm 5.13.2]{Rice}) that $q$ is constant, and therefore that every minimal Herglotz function is of the form
\begin{equation*}\label{mAgain}
\frac{p(z) + \sqrt{R(z)}}{a(z)}.
\end{equation*}
Similarly (\cite[Thm 11.7.10]{OPUC2}), minimal Carath\'{e}odory functions are of the same form. 
\end{remarks}

There is a one-to-one correspondence between all minimal Herglotz (Carath\'{e}odory) functions and the configurations of their poles (the so-called ``Dirichlet data''). Let us label the preimages of gaps under $\pi$ as follows:
\begin{equation*}\label{gap}
G_j := \pi^{-1}([\beta_j,\alpha_{j+1}]),\quad j=1,\ldots,l
\end{equation*}
for $\calS_\fre$ and
\begin{equation}\label{gapF}
G_j := \pi^{-1}(\{e^{i\theta}: \theta_{2j}\le \theta \le \theta_{2j+1}\}), \quad j=1,\ldots, l
\end{equation}
for $\calS_\frf$, where we adopt the convention $\theta_{2l+1}:=\theta_1+2\pi$. Note that each $G_j$ is topologically a circle, so that $\times_{j=1}^l G_j$ is an $l$-torus.

\begin{lemma}\label{lemMinimal}
\hspace*{\fill}\\
\indent {\normalfont $\rl$}
Every minimal Herglotz function has $l$ finite poles, each simple, one on each $G_j$ $(j=1,\ldots,l)$, and the map from minimal Herglotz functions to its finite poles is one-to-one and onto $\times_{j=1}^l G_j$.

\smallskip{\normalfont $\uc$}
Every minimal Carath\'{e}odory function has its $l$ poles, each simple, one on each $G_j$ $(j=1,\ldots,l)$, and the map from minimal Carath\'{e}odory functions to its poles is one-to-one and onto $\times_{j=1}^l G_j$.
\end{lemma}
\begin{remarks}
1. By the definition, each minimal Herglotz function has also a pole at $\infty_-$, see Def.~\ref{defMinimal}(iii). This makes the total of $l+1$ poles on $\calS_\fre$, which agrees with Def.~\ref{defMinimal}(ii).

2. Just like here, we will be looking to first classify all the Herglotz and Carath\'{e}odory functions of our finite range perturbations (Theorem~\ref{thmM}), and then we will classify their poles (Theorem~\ref{thmR}).
\end{remarks}

\subsection{Resonances}

If $m$ is meromorphic on the whole surface $\calS_\fre$ (we no longer assume $m$ is minimal), then resonances of $\calJ$ are defined to be the poles of $m$  on $\calS_{\fre,-}\setminus\{\infty_-\}$. At a band edges one should be more careful: a pole of $m$ of order $2$ is an eigenvalue of $\calJ$, while a pole of order $1$ is a resonance. The notion of resonances for $\calC$ and $F$ is analogous. Finally, the resonances of $\calJ$ on $\bbR$ (resp., $\calC$ on $\partial\bbD$) will be referred to as anti-bound states. Both eigenvalues and resonances of $\calJ$ or $\calC$ will be called singularities of $\calJ$ or $\calC$.

\subsection{Isospectral tori}

Using the relations~\eqref{stripping} (or~\eqref{strCar}), it is easy to see that the $m$-functions (Carath\'{e}odory functions) of periodic operators satisfy some quadratic equations. It should not be too surprising then that the solutions of this quadratic equation has a meromorphic continuation to a two-sheeted Riemann surface. In fact, $m$ (respectively, $F$) is a minimal Herglotz (Carath\'{e}odory) function, and conversely every minimal Herglotz (Carath\'{e}odory) function is an $m$-function (Carath\'{e}odory function) of some periodic Jacobi (CMV) operator with $\sigma_{ess}(\calJ)=\fre$ ($\sigma_{ess}(\calC)=\frf$). This explains the motivation behind the following definition.

\begin{definition}\label{defTorus}
\hspace*{\fill}\\
\indent {\normalfont $\rl$}
The isospectral torus $\calT_\fre$ is defined to be all the Jacobi operators whose $m$-functions are minimal Herglotz functions on $\calS_\fre$.

\smallskip{\normalfont $\uc$}
The isospectral torus $\calT_\frf$ is defined to be all the CMV operators whose Carath\'{e}odory functions are minimal Ca\-ra\-th\'{e}o\-do\-ry functions on $\calS_\fre$.
\end{definition}
\begin{remarks}
1. When $l=0$ the isospectral torus consists of one operator. If $\fre = [-2,2]$, then $\calT_\fre$ is the free Jacobi operator ($a_n=1, b_n=0$ for all $n$), and if $\frf = \partial\bbD$ then $\calT_\frf$ is the free CMV operator ($\alpha_n=0$ for all $n$).

2. When the harmonic measures of each band of $\fre$ are rational (see Subsection~\ref{ssPeriodic}) then $\calT_\fre$ consists precisely of all the periodic Jacobi operators with the essential spectrum $\fre$.  Similarly for $\frf$, but now the CMV operators could be periodic up to a phase (see Subsection~\ref{ssPeriodic}).

3. In general $\calT_\fre$ and $\calT_\frf$ consist of almost periodic operators with at most $l$ quasiperiods.

4. As was shown in~\cite{CSZ1,KruSim14}, if one extends the Jacobi matrices from $\calT_\fre$ to two-sided matrices, then an equivalent description of $\calT_\fre$ could be: (a) all the two-sided Jacobi matrices with $\sigma(\calJ)=\fre$ that are reflectionless on $\fre$ or (b) all the two-sided Jacobi operators with $\sigma(\calJ)=\fre$ that are almost periodic and regular.
\end{remarks}

\subsection{Finite range perturbations}
Let us classify all of the finite range perturbations by the number of the ``wrong'' (that is, not ``almost periodic'') coefficients. Note that if $\calJ^{(k)} = \calJ_\circ^{(k)}$ and $\calJ^{(k-1)} \ne \calJ_\circ^{(k-1)}$ then either the $b_k$ coefficients or the $a_k$ coefficients (or both) of $\calJ$ and $\calJ_\circ$ differ. This can be captured by saying that $\ran\big[\calJ^{(s-1)}-\calJ_\circ^{(s-1)}\big]$ is $1$ (which means $a_k$'s agree but $b_k$'s differ) or $2$ ($a_k$'s differ, possibly $b_k$'s too).

\begin{definition}\label{defFiniteRange}
\hspace*{\fill}\\
\indent {\normalfont $\rl$}
Let $\calT_\fre$ be the isospectral torus of Jacobi operators associated with a finite gap set $\fre$.
\begin{itemize}
\item Denote by $\calT_\fre^{[2s-1]}$ the set of all Jacobi operators $\calJ$ for which there exists $\calJ_\circ\in\calT_\fre$ such that $\calJ^{(s)}=\calJ_\circ^{(s)}$ and $\ran\big[\calJ^{(s-1)}-\calJ_\circ^{(s-1)}\big] = 1$;
\item Denote by $\calT_\fre^{[2s]}$ the set of all Jacobi operators $\calJ$ for which there exists $\calJ^\circ\in\calT_\fre$ such that $\calJ^{(s)}=\calJ^\circ$ and $\ran\big[\calJ^{(s-1)}-\calJ_\circ^{(s-1)}\big] = 2$.
\end{itemize}

\smallskip{\normalfont $\uc$}
Let $\calT_\frf$ be the isospectral torus of CMV operators associated with a finite gap set $\frf$.
\begin{itemize}
\item Denote by $\calT_\frf^{[s]}$ the set of all CMV operators $\calC$  for which $\calC^{(s)} \in \calT_\frf$, $\calJ^{(s-1)}\notin\calT_\frf$.
\end{itemize}

\end{definition}
\begin{remarks}
1. Both $\{a_j\}$ and $\{b_j\}$ sequences are \textit{eventually} almost periodic. One should think of the of the index $k$ in $\calT_\fre^{[k]}$ as the smallest number such that deleting the \textit{first} $k$ coefficients from the sequence $b_1,a_1,b_2,a_2,\ldots$ makes it almost periodic. It is important to put the $b$'s coefficients before the $a$'s here.


2. Thus the set of all finite range perturbations of $\calT_\fre$ splits into the disjoint union $\calT_\fre^{[0]}\cup \calT_\fre^{[1]} \cup \ldots \cup \calT_\fre^{[k]} \cup\ldots$, where $\calT_\fre^{[0]}$ is just $\calT_\fre$. Similarly for $\calT_\frf$.
\end{remarks}

\section{Classification of $m$-functions and Carath\'{e}odory functions}\label{sMandF}

Let us prove an easy lemma first.

\begin{lemma}\label{lemEqualDegree}
\hspace*{\fill}\\
\indent {\normalfont $\rl$}
Let
\begin{equation}\label{nonminimal1}
m(z)=\frac{p(z)+\sqrt{R_\fre(z)}}{a(z)}
\end{equation}
for some polynomials $p,a$ $(a\not\equiv0)$. The following three conditions are equivalent:
\begin{itemize}
\item[$(1)$] $\deg m = \Deg a$;
\item[$(2)$]
\begin{itemize}
\item[(a)] 
$m(z_0)=\infty$ for some $z_0\in \pi^{-1}(\bbC\setminus\cup_{j=1}^{l+1} \{\alpha_j,\beta_j\})$ implies $m^\sharp(z_0)\neq \infty$;
\item[(b)] If $z_0\in\pi^{-1}(\cup_{j=1}^{l+1} \{\alpha_j,\beta_j\})$ then it is at most simple pole of $m$.
\item[(c)] $\infty_+$ and $\infty_-$ are not poles of $m$.
\end{itemize}
\item[$(3)$]
\begin{itemize}
\item[(a)] 
$m(z_0)=\infty$ for some $z_0\in \pi^{-1}(\bbC\setminus\cup_{j=1}^{l+1} \{\alpha_j,\beta_j\})$ implies $m^\sharp(z_0)\neq \infty$;
\item[(b)] If $z_0\in\pi^{-1}(\cup_{j=1}^{l+1} \{\alpha_j,\beta_j\})$ then it is at most simple pole of $m$.
\item[(c)] $\Deg a \ge l+1$;
\item[(d)] $\Deg p \le \Deg a$;
\end{itemize}
\end{itemize}

\smallskip{\normalfont $\uc$}
Let
\begin{equation}\label{nonminimalF}
F(z)=\frac{p(z)+z^k\sqrt{R_\frf(z)}}{a(z)}
\end{equation}
for some integer $k\ge 0$ and some polynomials $p,a$ with $a(0)\ne 0$. 
The following three conditions are equivalent:
\begin{itemize}
\item[$(1)$] $\deg F = \Deg a$;
\item[$(2)$]
\begin{itemize}
\item[(a)] 
$F(z_0)=\infty$ for some $z_0\notin \pi^{-1}(\{0\}\cup_{j=1}^{2l} \{e^{i\theta_j}\})$ implies $F^\sharp(z_0)\neq \infty$.
\item[(b)] If $z_0\in\pi^{-1}(\cup_{j=1}^{2l} \{e^{i\theta_j}\})$ then it is at most simple pole of $F$.
\item[(c)] $\infty_+$ and $\infty_-$ are not poles of $F$.
\end{itemize}
\item[$(3)$]
\begin{itemize}
\item[(a)] 
$F(z_0)=\infty$ for some $z_0\notin \pi^{-1}(\{0\}\cup_{j=1}^{2l} \{e^{i\theta_j}\})$ implies $F^\sharp(z_0)\neq \infty$.
\item[(b)] If $z_0\in\pi^{-1}(\cup_{j=1}^{2l} \{e^{i\theta_j}\})$ then it is at most simple pole of $F$.
\item[(c)] $\Deg a \ge l+k$;
\item[(d)] $\Deg p \le \Deg a$;
\end{itemize}
\end{itemize}
\end{lemma}
\begin{proof}
$\rl$
First of all, note that $\deg m \ge \Deg a$ always holds. To see that, let us compare the number of zeros of $a$ and finite poles of $m$. Indeed, if $z_0\in\bbC\setminus\cup_{j=1}^{l+1}\{\alpha_j,\beta_j\}$ is a zero of $a$ of order $n$, then at least one of $(z_0)_+$ or $(z_0)_-$ will be a pole of $m$ of order $n$ since $\sqrt{R_\fre(z_0)}\ne 0$. If $a$ has a zero at an endpoint $z_0\in\cup_{j=1}^{l+1}\{\alpha_j,\beta_j\}$ of order $1$, then $m$ has a pole of order $1$ at $(z_0)_+=(z_0)_-$ if $p(z_0) = 0$, or $m$ has a pole of order $2$ if $p(z_0)\ne 0$. Moreover, if $a$ has a zero of order $n\ge 2$ at an endpoint, then $m$ blows up at least as $(z-z_0)^{n-1/2}$, i.e., has a pole of order $\ge 2n-1 > n$.

$(1)\Rightarrow(2)$ From the above considerations, in order for the equality in $\deg m \ge \Deg a$ to hold, we must have (2a) and (2b). Moreover, a pole of $m$ at $\infty_\pm$ would also break the equality, so (2c) must hold too.

$(2)\Rightarrow(3)$ Note that $\sqrt{R_\fre(z)} \sim \pm z^{l+1}$ at $z\to\infty_\pm$. This means that (2c) requires $\Deg a \ge l+1$, and then $\Deg p \le \Deg a$.

$(3)\Rightarrow(1)$ Conditions (3c) and (3d) imply that $\infty_+$ and $\infty_-$ are not poles of $m$. Therefore all the poles of $m$ come from the zeros of $a$. The arguments in the beginning of the proof show that (3a) and (3b) guarantee that the total number of zeros of $a$ and poles of $m$ coincide when counted  with their multiplicities.

\smallskip{\normalfont $\uc$}
The proof for $F$ follows along the identical lines. Note that we do not need to worry about points $0_\pm$ since we are assuming in advance that $a(0)\ne 0$, and therefore $0_\pm$ cannot be poles of $F$.
\end{proof}

\begin{lemma}\label{lemInduct}
\hspace*{\fill}\\
\indent {\normalfont $\rl$}
Suppose two functions $m$ and $m^{(1)}$ meromorphic on $\calS_\fre$ satisfy~\eqref{stripping}
for some constants $a_1\ne0, b_1\in\bbC$.
\begin{itemize}
\item[(i)]
If $m$ is of the form~\eqref{nonminimal1} and satisfies $(2a)$ and $(2b)$ of Lemma~\ref{lemEqualDegree}, then the same is true of $m^{(1)}$.
\item[(ii)]
If $m^{(1)}$ is of the form~\eqref{nonminimal1} and satisfies $(2a)$ and $(2b)$ of Lemma~\ref{lemEqualDegree}, then the same is true of $m$.
\end{itemize}

\smallskip{\normalfont $\uc$}
Suppose two functions $F$ and $F^{(1)}$ meromorphic on $\calS_\frf$ satisfy~\eqref{strCar} for some $\alpha_0\in\bbD$.
\begin{itemize}
\item[(i)]
If $F$ is of the form~\eqref{nonminimalF} with $a(0)\ne0$ and satisfies $(2a)$ and $(2b)$ of Lemma~\ref{lemEqualDegree}, then $F^{(1)}$ is of the form~\eqref{nonminimalF} with a possibly different $k$ and satisfies $(2a)$ and $(2b)$ of Lemma~\ref{lemEqualDegree}.
\item[(ii)]
If $F^{(1)}$ is of the form~\eqref{nonminimalF} with $a^{(1)}(0)\ne0$  and satisfies $(2a)$ and $(2b)$ of Lemma~\ref{lemEqualDegree}, then $F$ is of the form~\eqref{nonminimalF} with a possibly different $k$ and satisfies $(2a)$ and $(2b)$ of Lemma~\ref{lemEqualDegree}.
\end{itemize}
\end{lemma}
\begin{remark}
In \uc{}(i) we do not claim that necessarily $a^{(1)}(0)\ne 0$. In \uc{}(ii) we do not claim that necessarily $a(0)\ne 0$. This will be automatic later on when we know that $F$ and $F^{(1)}$ are Carath\'{e}odory functions.
\end{remark}
\begin{proof}
\rl{}
First of all, note that for any function of the form~\eqref{nonminimal1} that satisfies (2a) and (2b), $\frac{p(z)^2-R_\fre(z)}{a(z)}$ is a polynomial. Indeed, let $z_0$ be a zero of $a(z)$ of order $n$. If $z_0\in\bbC\setminus\cup_{j=1}^{l+1}\{\alpha_j,\beta_j\}$, then (2a) implies that $p(z) + \sqrt{R_\fre(z)}$ 
must have a zero of order $\ge n$ at $(z_0)_+$ or at $(z_0)_-$. This implies that the polynomial $p(z)^2-R_\fre(z)$ has a zero of order $\ge n$ at $z_0$. If $z_0\in \cup_{j=1}^{l+1}\{\alpha_j,\beta_j\}$ then (2b) and the arguments in the proof of Lemma~\ref{lemEqualDegree} imply that $n=1$ and $p(z_0)=0$ which shows that $p(z)^2-R_\fre(z)$ has a zero at $z_0$. This proves that $p(z)^2-R_\fre(z)$ is divisible by $a(z)$.

(i) Plugging~\eqref{nonminimal1} into~\eqref{stripping}, simplifying, and using the divisibility of $p(z)^2-R_\fre(z)$ by $a(z)$, one can see that $m^{(1)}$ is indeed of the form~\eqref{nonminimal1}. Suppose that $m^{(1)}$ has a pole at $z_0$ and at $\tau(z_0)$ for some $z_0\in\pi^{-1}(\bbC)$. By~\eqref{stripping} this implies $m(z_0)=m(\tau(z_0))$. But $m$ is of the form~\eqref{nonminimal1} which implies $R_\fre(z_0) = 0$, i.e., $z_0\in\pi^{-1}(\cup_{j=1}^{l+1} \{\alpha_j,\beta_j\})$ which establishes (2a) for $m^{(1)}$. Now suppose $z_0\in\pi^{-1}(\cup_{j=1}^{l+1} \{\alpha_j,\beta_j\})$ is a pole of order $n$ for $m^{(1)}$. Then $m$ has a zero of order $n$ at $z_0$ by~\eqref{stripping}. By~\eqref{nonminimal1} $m$ has a zero at $z_0$ if and only if $a(z_0)\ne 0$ and $p(z_0)=0$, but then $m\sim \operatorname{const} \times (z-z_0)^{1/2}$ as $z\to z_0$, so $n=1$.

(ii) Plugging
\begin{equation}\label{mStripped}
m^{(1)} = \frac{p^{(1)}(z)+ \sqrt{R_\fre(z)}}{a^{(1)}(z)}
\end{equation}
into~\eqref{stripping} and using that $\frac{p^{(1)}(z)^2-R_\fre(z)}{a^{(1)}(z)}$ is a polynomial, one can see that $m$ is indeed of the form~\eqref{nonminimal1}. Suppose that $m$ has a pole at $z_0$ and at $\tau(z_0)$ for some $z_0\in\pi^{-1}(\bbC)$. By~\eqref{stripping} this implies $m^{(1)}(z_0)=m^{(1)}(\tau(z_0))$. By~\eqref{mStripped} this implies $R_\fre(z_0) = 0$, i.e., $z_0\in\pi^{-1}(\cup_{j=1}^{l+1} \{\alpha_j,\beta_j\})$ which establishes (2a) for $m$.

Finally suppose $z_0\in\pi^{-1}(\cup_{j=1}^{l+1} \{\alpha_j,\beta_j\})$ is a pole of order $n$ for $m$.
Let $m(z)=\sum_{j=-n}^\infty c_j (z-z_0)^{j/2}$ be the Taylor expansion of $m$ around $z_0$. Then $m^\sharp(z)=\sum_{j=-n}^\infty (-1)^j c_j (z-z_0)^{j/2}$. 
Thus $m^\sharp$ has a pole of order $n$ at $z_0$ as well, and the order of the pole of $m-m^\sharp$ at $z_0$ is either $n$ or $n-1$. Note that~\eqref{stripping} implies
$$
a_1^2 \left(m^{(1)}(z) -m^{(1)}{}^\sharp(z)\right)= \frac{m(z)-m^\sharp(z)}{m(z) m^\sharp(z)}.
$$
Therefore $m^{(1)}-m^{(1)}{}^\sharp$ has a zero at $z_0$ of order $\ge 2n-n= n$. But as argued in (i), a function of the form~\eqref{mStripped} cannot have a zero at $\pi^{-1}(\cup_{j=1}^{l+1} \{\alpha_j,\beta_j\})$ of order higher than $1$. This shows $n=1$.

\smallskip{\normalfont $\uc$}
Using similar arguments as above, for any function of the form~\eqref{nonminimalF} with $a(0)\ne 0$ that satisfies (2a) and (2b), $\frac{p(z)^2-z^{2k} R_\frf(z)}{a(z)}$ is a polynomial.

(i) Plugging~\eqref{nonminimalF} into~\eqref{strCar} and using the divisibility of $p(z)^2-z^{2k} R_\frf(z)$ by $a(z)$, tedious but straightforward computations show that $F^{(1)}$ is indeed of the form~\eqref{nonminimalF} with a possibly different $k$. Suppose that $F^{(1)}$ has a pole at $z_0$ and at $\tau(z_0)$ for some $z_0\in\pi^{-1}(\bbC\setminus\{0\})$. Rewrite~\eqref{strCar} as
\begin{equation}\label{strCar2}
 \frac{F(z)+1}{F(z)-1} = \frac{1}{z} \, \frac{(F^{(1)}(z)+1) + \bar\alpha_0(F^{(1)}(z)-1)}{\alpha_0(F^{(1)}(z)+1) + (F^{(1)}(z)-1)}.
\end{equation}
This shows that $F^{(1)}(z_0) = F^{(1)}(\tau(z_0)) = \infty$ implies $F(z_0)=F(\tau(z_0))$. $F(z_0)=F(\tau(z_0))=\infty$ is impossible since we are assuming $F$ satisfies (2a), and otherwise $F(z_0)=F(\tau(z_0))$ implies $z_0^k \sqrt{R_\frf(z_0)} = 0$, i.e., $z_0\in\pi^{-1}(\{0\}\cup_{j=1}^{2l} \{e^{i\theta_j}\})$. This proves (2a) for $F^{(1)}$.

Now suppose $z_0\in\pi^{-1}(\cup_{j=1}^{2l} \{e^{i\theta_j}\})$ is a pole of order $n\ge 2$ for $F^{(1)}$. Since any M\"{o}bius  transformation is conformal, the right-hand side of~\eqref{strCar2} as $z\to z_0$ takes the form
\begin{equation}\label{indEq1}
\frac{1}{z} \left(c_0 + c_n (z-z_0)^{n/2} + o((z-z_0)^{n/2}\right)
\end{equation}
with $c_n\ne 0$. Since $\tfrac{1}{z} = \tfrac{1}{z_0} + O((z-z_0))$, ~\eqref{indEq1} becomes $\tfrac{c_0}{z_0} + O((z-z_0))$. Notice the absence of $(z-z_0)^{1/2}$ term! Then using~\eqref{strCar2} and  conformality of a M\"{o}bius transformation again, we obtain that $F$ has a pole of order $\ge 2$ at $z_0$ or $F(z)-F(z_0)$ has a zero at $z_0$ of order $\ge2$. The former case is impossible since $F$ satisfies (2b), while the latter case is also impossible since we get $a(z_0)\ne 0$, and the $(z-z_0)^{1/2}$ term in the Taylor's series $z^k \sqrt{R_\frf(z)} = c_1 (z-z_0)^{1/2} + O((z-z_0))$  ($c_1\ne 0$) cannot be canceled by $p(z) = p(z_0) + O((z-z_0))$. We got a contradiction with $n\ge 2$ and therefore proved (2b) for $F^{(1)}$.

(ii) Plugging
\begin{equation}\label{FStripped}
F^{(1)} = \frac{p^{(1)}(z)+ z^k \sqrt{R_\frf(z)}}{a^{(1)}(z)}
\end{equation}
into~\eqref{strCar2} and using that $\frac{p^{(1)}(z)^2-z^{2k} R_\frf(z)}{a^{(1)}(z)}$ is a polynomial, one obtains that $F$ is indeed of the form~\eqref{nonminimalF}. Suppose that $F$ has a pole at $z_0$ and at $\tau(z_0)$ for some $z_0\in\pi^{-1}(\bbC)$. By~\eqref{strCar} this implies $F^{(1)}(z_0)=F^{(1)}(\tau(z_0))$ (note that the ratio on the right-hand side is never $\tfrac{0}{0}$ since $|\alpha_0|\ne 1$). $z_0$ and $\tau(z_0)$ cannot be poles of $F^{(1)}$ by (2a), and then $F^{(1)}(z_0)=F^{(1)}(\tau(z_0))$ with~\eqref{FStripped} implies $z_0^k R_\frf(z_0) = 0$, i.e., $z_0\in\pi^{-1}(\{0\}\cup_{j=1}^{2l} \{e^{i\theta_j}\})$ which establishes (2a) for $F$.

Finally suppose $z_0\in\pi^{-1}(\cup_{j=1}^{2l} \{e^{i\theta_j}\})$ is a pole of order $n\ge2$ for $F$. Then reusing the conformality arguments in (i) we can see that $\frac{F(z)+1}{F(z)-1} = 1 + O((z-z_0))$, $z\to z_0$. Then $z\frac{F(z)+1}{F(z)-1} = z_0 + O((z-z_0))$ (the stress is on the absence of $(z-z_0)^{1/2}$ term), and then by~\eqref{strCar} we get that $F^{(1)}$ must has a pole of order $\ge 2$ at $z_0$ or $F^{(1)}(z)-F^{(1)}(z_0)$ has a zero of order $2$. The first case is impossible by (2b), and the second is impossible due to the presence of $\sqrt{R_\frf(z)} = c_1 (z-z_0)^{1/2} + O((z-z_0))$, $c_1\ne 0$, just like in the proof of (i).
\end{proof}

For a future reference, we note that if the $m$-function of $\mu$ is of the form~\eqref{nonminimal1}, then using~\eqref{symmM} we can rewrite~\eqref{herg1} as
\begin{equation}\label{densityM}
\frac{d\mu}{dx} = \tfrac{1}{2\pi i} \lim_{\veps\downarrow 0} m((x+i\veps)_+)-m((x-i\veps)_+) = \frac{m(x_+)-m(x_-)}{2\pi i}  = \frac{\sqrt{R_\fre(x)}}{\pi i a(x)}
\end{equation}
for $x\in\fre$.
Similarly, if the Carath\'{e}odory function $F$ of $\mu$ is of the form~\eqref{nonminimalF}, then using~\eqref{symmF}, we can rewrite~\eqref{cara1} as
\begin{equation}\label{densityF}
\frac{d\mu}{d\theta} = \tfrac{1}{4\pi} \lim_{r\uparrow 1} F((re^{i\theta})_+)-F((r^{-1} e^{i\theta})_+) = \frac{F((e^{i\theta})_+)-F((e^{i\theta})_-)}{4\pi}  = \frac{e^{ik\theta } \sqrt{R_\frf(e^{i\theta})}}{2\pi a(e^{i\theta})}.
\end{equation}
for $\theta\in\frf$.

We are now ready to prove the classification of the $m$-functions (Carath\'{e}odory functions) of finite range perturbations.

\begin{theorem}\label{thmM}
\hspace*{\fill}\\
\indent {\normalfont $\rl$}
Let $\calJ[a_n,b_n]_{n=1}^\infty$ be a Jacobi operator and $m$ its $m$-function~\eqref{m}. The following are equivalent:
\begin{itemize}
\item[$(\calT_\fre^{s\ge 1})$] $\calJ$ is in $\calT_\fre^{s\ge 1}$.
\item[$(\calM_\fre^{s\ge1})$] The $m$-function of $\calJ$ is of the form
\begin{equation}\label{simpleM}
m(z)=\frac{p(z)+\sqrt{R_\fre(z)}}{a(z)},
\end{equation}
where $p,a$ are polynomials and
\begin{equation}\label{degsEqual}
\deg m = \Deg a.
\end{equation}
\end{itemize}
Moreover, for any $s\ge1$,
\begin{equation}\label{singularitiesNumber}
\calJ\in\calT_\fre^{[s]} \mbox{ if and only if } \deg m = \Deg a = l + s.
\end{equation}

\smallskip{\normalfont $\uc$}
Let $\calC[\alpha_n]_{n=0}^\infty$ be a CMV operator and $F$ its Carath\'{e}odory function~\eqref{F}. The following are equivalent:
\begin{itemize}
\item[$(\calT_\frf^{s\ge 0})$] $\calC$ is in $\calT_\frf^{s\ge 0}$.
\item[$(\calM_\frf^{s\ge 0})$] The Carath\'{e}odory function  of $\calC$ is of the form
\begin{equation}\label{simpleF}
F(z)=\frac{p(z)+z^s \sqrt{R_\frf(z)}}{a(z)},
\end{equation}
for some $s\ge 0$, where $p,a$ are polynomials with $a(0)\ne 0$ and
\begin{equation}\label{degsEqualF}
\deg F = \Deg a = l+2s.
\end{equation}
\end{itemize}
Moreover, for any $s\ge0$,
\begin{equation}\label{singularitiesNumberF}
\calC\in\calT_\frf^{[s]} \mbox{ if and only if } \deg F = \Deg a = l + 2s.
\end{equation}
\end{theorem}
\begin{remarks}
1. Note that $s\ne 0$ in the condition $(\calT_\fre^{s\ge 1})$. Indeed, for the Jacobi operators in the isospectral torus we have in fact $\deg m = l+1$, $\Deg a=l$.

2. See Lemma~\ref{lemEqualDegree} for an intuition on what $\deg m = \Deg a$ (respectively, $\deg F = \Deg a$) for such functions means.

3. As we show later, all such functions $m$ and $F$ are uniquely determined by the set of their poles. In Section~\ref{sResonance} we show the necessary and sufficient condition for any set of points on $\calS_\fre$ to be the set of poles of such a function. Given such a configuration, we present an explicit form of $m$ in Section~\ref{sInterpolation}.
\end{remarks}
\begin{proof}
\rl{}

$(\calT_\fre^{s\ge 1})\Rightarrow(\calM_\fre^{s\ge 1})$

The $m$-function of any operator in $\calT_\fre$ is a minimal Herglotz function, see Def.~\ref{defMinimal}. In particular it is of the form~\eqref{simpleM} (see Remark 4 after Def.~\ref{defMinimal}) and satisfies (2a) and (2b) of Lemma~\ref{lemEqualDegree} (follows from Lemma~\ref{lemMinimal}). By Lemma~\ref{lemInduct} the same is true of any operator in $\calT_\fre^{[s]}$ for all $s$. So we just need to establish (2c) and~\eqref{singularitiesNumber}.


Suppose $\calJ\in\calT_\fre^{[s]}$ with $s=1$ or $s=2$, i.e., $\calJ^{(1)}\in \calT_\fre$. Let $m^{(1)}$ be the $m$-function of $\calJ^{(1)}$. As a minimal Herglotz function, $m^{(1)}(z)$ has exactly one pole per gap $\pi^{-1}([\alpha_j,\beta_j])$ and a first order pole at $\infty_-$  (see Lemma~\ref{lemMinimal} and the remark following it). Let $m^{(1)}(z)\sim k_{1}z+k_0+O(\tfrac1z)$, $k_1\ne0$, at $\infty_-$. 

Let us rewrite~\eqref{stripping} as
\begin{equation}\label{stripping2}
m(z)=\frac{1}{b_1-z-a_1^2 m^{(1)}(z)}
\end{equation}
and count the solutions of the equation $m(z)=0$: exactly once per each gap (where $m^{(1)}$ has a pole), a simple zero at $\infty_+$ (since $m$ is an $m$-function), and possibly a simple zero at $\infty_-$. Note that $m(z)=0$ at $\infty_-$ if and only if $1+a_1^2 k_1 \ne 0$. But we know that if $\calJ\in\calT_\fre$ then $m(z)$ has a pole at $\infty_-$ as a minimal Herglotz function. Therefore $a_1=\sqrt{-1/k_1}$ is precisely the unique $a_1$ that makes the sequence $a_1,a_2,a_3\ldots$ almost periodic (indeed, we know there exists a unique such $a_1$, see Subsection~\ref{ssPeriodic}).

If $s=2$, then $a_1,a_2,a_3\ldots$ is not almost periodic, so $a_1 \ne \sqrt{-1/k_1}$, and therefore we just  showed that $m$ has exactly $l+2$ zeros: once per each gap, one at $\infty_+$, and one at $\infty_-$. In particular (2c) of Lemma~\ref{lemEqualDegree} for $m$ holds, and $\deg m = l+2$. Therefore $\deg m = \Deg a= l+2$ by Lemma~\ref{lemEqualDegree}.

If $s=1$, then $a_1 =\sqrt{-1/k_1}$, and so $\infty_-$ is not a zero of $m$. By the zero counting above, we have precisely $l+1$ zeros (one per each gap and $\infty_+$), i.e., $\deg m = l+1$. By~\eqref{stripping2},  $\infty_-$ is a pole of $m$ if and only if $a_1^2 k_0 -b_1=0$. Minimal Herglotz functions have a pole at $\infty_-$, which means that $b_1 =- k_0/k_1$ is exactly the condition for $m$ to be minimal (put it another way, the unique value of $b_1$ to make $b_1,b_2,b_3,\ldots$ to be almost periodic, which again, we know happens for a unique choice of $b_1$). But $s=1$, so $\calJ\notin \calT_\fre$. Thus $b_1\ne - k_0/k_1$, and $\infty_-$ is not a pole of $m$. This proves that $m$ satisfies (2c) of Lemma~\ref{lemEqualDegree}.

Now $s\ge 3$ follows easily by induction. Note that $\infty_-$ was not a pole of $m$ in either of the cases $s=1$ or $s=2$ above. Therefore by~\eqref{stripping}, $\infty_-$ is always a zero when $s\ge3$, so (2c) of Lemma~\ref{lemEqualDegree} applies. Using~\eqref{stripping} again we obtain that $m$ has zeros at $\infty_+$, at $\infty_-$, and at every pole of $m^{(1)}$. Therefore $\deg m = \deg m^{(1)}+2$.

Note that we proved the ``moreover'' part of the theorem along the way too.

\medskip
$(\calM_\fre^{s\ge 1})\Rightarrow  (\calT_\fre^{s\ge1})$
Suppose $m$ is of the form~\eqref{simpleM} and satisfies~\eqref{degsEqual}. Since it is the $m$-function of some $\calJ$, we can consider $\calJ^{(k)}$, $k\ge1$, and the corresponding $m$-functions $m^{(k)}$. By Lemmas~\ref{lemEqualDegree} and~\ref{lemInduct}, each $m^{(k)}$ is also of the form~\eqref{simpleM} (with $p^{(k)}$ and $a^{(k)}$ instead of $p$ and $a$) and satisfies (2a) and (2b) of Lemma~\ref{lemEqualDegree}.

Note that a function of the form~\eqref{simpleM} has $\deg m \ge l+1$ (see~\cite[Thm 5.12.5]{Rice}).
Let $\deg m =\Deg a  = l+s$ with $s\ge 1$.

Let us carefully check the configuration of zeros of $m$. There is a total of $l+s$ of zeros when counted with the multiplicities, and $\infty_+$ is one of them since $m\sim-\tfrac{1}{z}$, $z\to\infty_+$. We would like to know whether $\infty_-$ is also a zero. Note that
\begin{equation}\label{continuation}
m(z) = m^\sharp(z) + \frac{2\sqrt{R_\fre(z)}}{a(z)}.
\end{equation}
Since $\Deg a = l+s$ and $\Deg R_\fre = 2(l+1)$, this shows that if $s=1$ then $\infty_-$ is neither a zero nor a pole of $m$, and if $s\ge 3$ then $\infty_-$ is a simple zero of $m$. When $s=2$, we also obtain that $\infty_-$ is a zero of $m$, but we have to be more careful in order to justify that it is simple.
Recall~\eqref{densityM}. Since $i\sqrt{R_\fre(x)}$ changes sign from one band to another, and $\tfrac{d\mu(x)}{dx}\ge 0$, we obtain that $a$ is real on $\bbR$ and must have an odd number of zeros on each gap $[\beta_j,\alpha_{j+1}]$, counted with multiplicities. We claim that in each gap $\pi^{-1}([\beta_j,\alpha_{j+1}])$, $m$ must have at least one zero. Indeed, $m$ restricted to $\pi^{-1}([\beta_j,\alpha_{j+1}])$ is a smooth map from $\pi^{-1}([\beta_j,\alpha_{j+1}])$, homeomorphic to a circle, into $\bbR\cup\{\infty\}$, also a circle. Moreover, it attains $\infty$ an odd number of times (in fact, 1 or 3 times in this case), which implies that this $S^1\to S^1$ map has a nonzero winding number. Therefore $m$ must attain $0$ at least once in each gap $\pi^{-1}([\beta_j,\alpha_{j+1}])$. We showed that $m$ must have at least $l$ \textit{finite} zeros. Because of $\deg m = l+2$ and a simple zero at $\infty_+$, we conclude that the zero at $\infty_-$ is also simple.

To sum up, we showed that if $s=1$ then $m(\infty_-)\notin\{0,\infty\}$ and the zeros of $m$ are: a simple zero at $\infty_+$ and $l$ finite zeros (counted with multiplicities); and if $s\ge 2$ then the zeros of $m$ are: a simple zero at $\infty_+$, a simple zero at $\infty_-$, and $l+s-2$ of finite zeros (counted with multiplicities).

Consider now the case $s=1$. Let us count the poles of $m^{(1)}$. By~\eqref{stripping}, these occur at each of the \textit{finite} zeros of $m$ and possibly at $\infty_-$ (note that $\infty_+$ is never a pole since $m^{(1)}$ is an $m$-function). By the above, $m$ has $l$ finite zeros. At $\infty_-$, $m^{(1)}$ has a simple pole by $m(\infty_-)\notin\{0,\infty\}$ and~\eqref{stripping}. Therefore $\deg m^{(1)} = l+1$. One can recognize now that $m^{(1)}$ is a minimal Herglotz function (see Def.~\ref{defMinimal}), in other words, $\calJ^{(1)} \in \calT_\fre$.

If $s=2$, then poles of $m^{(1)}$ occur at each of the $l$ finite zeros of $m$, and possibly at $\infty_-$. Recall that $m(z)\sim -\tfrac{1}{z}$ at $\infty_+$, and therefore by~\eqref{continuation} we have $m(z) =  \tfrac{k}{z} + O(\tfrac{1}{z^2})$, $z\to\infty_-$, where $k\ne -1$ and $k\ne 0$ ($\infty_-$ is a simple zero). But then $\tfrac{1}{m(z)} + z = (\tfrac{1}{k}+1) z + O(1)$ with $\tfrac1k+1 \ne 0$. This and~\eqref{stripping} implies that $m^{(1)}$ has a simple pole at $\infty_-$. Moreover, $\deg m^{(1)} = l+1$ by counting its poles, which again means that $m^{(1)}$ is a minimal Herglotz function, that is, $\calJ^{(1)} \in \calT_\fre$.

Finally, suppose $s\ge 3$. By~\eqref{continuation}, $m(z)=-\tfrac{1}{z}+O(\tfrac{1}{z^2})$ at $\infty_-$, which implies $\tfrac{1}{m(z)} + z = O(1)$. Then~\eqref{stripping} shows that $\infty_-$ is not a pole of $m^{(1)}$. Therefore $m^{(1)}$ satisfies (2a), (2b), (2c) of Lemma~\ref{lemEqualDegree} with $\deg m^{(1)} = \deg m - 2$. Indeed, the poles of $m^{(1)}$ occur only at the finite zeros of $m$, and there are $\deg m-2$ of them. An induction completes the proof.


\smallskip{\normalfont $\uc$}

$(\calT_\frf^{s\ge 0})\Rightarrow(\calM_\frf^{s\ge 0})$

The Carath\'{e}odory function $F$ of any operator in $\calT_\frf$ is a minimal Carath\'{e}odory function, see Def.~\ref{defMinimal}. In particular, it is of the form~\eqref{simpleF} with $s=0$ (see Remark 4 after Def.~\ref{defMinimal}), satisfies $a(0)\ne0$ (since $F(0_+)=1$) and~\eqref{degsEqualF} (follows from Lemma~\ref{lemMinimal}). This shows $s=0$ case. Moreover, by Lemma~\ref{lemInduct}, the Carath\'{e}odory function of any operator in $\calT_\frf^{[s]}$ ($s\ge 1$) is of the form~\eqref{nonminimalF} for some $k$, and satisfies (2a) and (2b) of Lemma~\ref{lemEqualDegree}. So we just need to establish (2c) to be able to apply Lemma~\ref{lemEqualDegree}, and also show that $k=s$ and~\eqref{singularitiesNumberF}.

Suppose $\calC[\alpha_n]_{n=0}^\infty\in\calT_\frf^{[s]}$ with $s=1$, i.e., $\calC^{(1)}\in \calT_\frf$, $\calC\notin \calT_\frf$. Let us denote $\gamma^\circ$ to be the unique complex number that makes $\calC^\circ:=\calC[\gamma^\circ, \alpha_1,\alpha_2,\ldots]\in\calT_\frf$. Since $s=1$, we know that $\alpha_0\ne \gamma^\circ$. Let $F^{(1)}$ and $F^\circ$ be the Carath\'{e}odory functions  of $\calC^{(1)}$ and $\calC^\circ$. Both of them are minimal Carath\'{e}odory functions. Denote
\begin{align}
\label{formF1} F^{(1)}(z) & = \frac{p^{(1)}(z) + \sqrt{R_\frf(z)}}{a^{(1)}(z)}, \\
 F^{\circ}(z) & = \frac{p^{\circ}(z) + \sqrt{R_\frf(z)}}{a^{\circ}(z)}.
\end{align}
We  claim
\begin{align}
\label{value1} F^{(1)}(0_+) &= 1; \quad F^{(1)}(\infty_+) = -1; \\
\label{value2} F^{(1)}(0_-) &= \frac{\bar\gamma^\circ - 1}{\bar\gamma^\circ + 1}; \quad
F^{(1)}(\infty_-) = \frac{1-\gamma^\circ}{1+\gamma^\circ}.
\end{align}
Indeed,~\eqref{value1} follows from the definition~\eqref{F} and~\eqref{symmF}. To show~\eqref{value2},
first notice that
$$
F^{\circ}(z) = F^{\circ}{}^\sharp(z) + \frac{2 \sqrt{ R_\frf(z)}}{a^{\circ}(z)},
$$
which at $z=0_-$ gives us $F^{\circ}(0_-) \ne F^{\circ}(0_+) = 1$ since $a^{\circ}(0)\ne0$ and $R_\frf(0) \ne 0$. Now, since $\calC^\circ{}^{[1]} = \calC^{[1]}$, we can apply~\eqref{strCar2} with $F^\circ$ and $\gamma^\circ$ instead of $F$ and $\alpha_0$, respectively, and take a limit as $z\to 0_-$. The left-hand side is a finite number since $F^{\circ}(0_-)\ne 1$, which means that the numerator of the right-hand side must be zero, producing the first equality in~\eqref{value2}. The second equality in~\eqref{value2} follows from~\eqref{symmF} (by analytic continuation it holds on $\calS_{\frf,-}$ also).

Since we will be counting poles with their multiplicities, in what follows let us use $P(z_0,f)\in\{0,1,2,\ldots\}$ to denote the order of $z_0\in\calS_\frf$ as a pole of a function $f$.

Define a function $g$ to be the right-hand side of~\eqref{strCar2} multiplied by $z$ (that is, RHS of~\eqref{strCar2} $=\tfrac1z g(z)$). Note that $g$ is a composition of $F^{(1)}$ with two M\"{o}bius transformations. Since M\"{o}bius transformations are bijective and conformal on the Riemann sphere, $g$ is a meromorphic function on $\calS_\frf$ whose degree is equal to $\deg F^{(1)} = l$.
From the definition of $g$ and~\eqref{value1},~\eqref{value2} we get
\begin{align}
\label{value5} g(0_+) &= \frac{1}{\alpha_0}; \quad g(\infty_+) = \bar\alpha_0; \\
\label{value6} g(0_-) &= \frac{ \bar\alpha_0-\bar\gamma^\circ}{1-\alpha_0\bar\gamma^\circ }; \quad
g(\infty_-) = \frac{1-\bar\alpha_0\gamma^\circ}{\alpha_0-\gamma^\circ}.
\end{align}
(this is true even if $\alpha_0=0$ or if $\gamma^\circ = 0$). Let us count the poles of $\tfrac1z g(z)$. Since $\gamma^{\circ}\ne \alpha_0$, we have that $\infty_\pm$ are never poles of $g$ or of $\tfrac{1}{z} g(z)$. It is also clear that $P(z_0,\tfrac1z g) = P(z_0,g)$ for any $z_0\in\calS_\frf\setminus\{\infty_\pm,0_\pm\}$.
Note that $g(0_+) \ne 0$, as well as $g(0_-)\ne 0$ since $\gamma^{\circ}\ne \alpha_0$. This means that $P(0_+,\tfrac1z g) = P(0_+,g)+1$, $P(0_-,\tfrac1z g) = P(0_-,g)+1$.
This proves that $\deg \tfrac{1}{z} g(z) = l+2$. By applying a M\"{o}bius transformation in~\eqref{strCar2}, we can see that $\deg F = l+2$.

%
%

Plugging in~\eqref{value5} and~\eqref{value6} into~\eqref{strCar2} we can see that assuming $\alpha_0\ne \gamma^0$ we always have
\begin{align}
\label{value7} F(0_+) &= 1; \quad F(\infty_+) = -1; \\
\label{value8} F(0_-) &= 1; \quad
F(\infty_-) = -1.
\end{align}
Therefore we can apply Lemma~\ref{lemEqualDegree} to conclude $\deg F = \Deg a = l+2$. We are just left to show that $F$ is of the form~\eqref{simpleF} with $s=1$. For this, let us use~\eqref{strCar2}, solve for $F$ and then compute $F(z)-F^\sharp(z)$. After all the unsightly computations we end up with
\begin{equation}\label{ffsharp}
F(z) - F^\sharp(z) = \frac{ 4z (1-|\alpha_0|^2) \left[ F^{(1)}(z) - F^{(1)}{}^\sharp(z) \right] }{(A(z)+B(z) F^{(1)}(z))(A(z)+B(z)F^{(1)}{}^\sharp(z))},
\end{equation}
where $A(z)=(1-\bar\alpha_0)+z(1-\alpha_0)$, $B(z)= (1+\bar\alpha_0)-z(1+\alpha_0)$. Using~\eqref{value1}, we get $\lim_{z\to 0_+} A(z)+B(z) F^{(1)}(z) = 2$, and $\lim_{z\to 0_+} A(z)+B(z) F^{(1)}{}^\sharp(z) = (1-\bar\alpha_0) + (1+\bar\alpha_0)\frac{\bar\gamma^\circ - 1}{\bar\gamma^\circ + 1}$. The latter limit is in $\bbC\setminus\{0\}$ since $\gamma^\circ\ne \alpha_0$. Using this and~\eqref{formF1}, we get
$$
\lim_{z\to 0_+} \frac{F(z) - F^\sharp(z)}{z} \in \bbC\setminus\{0\}
$$
which means $s=1$ in~\eqref{simpleF} and finishes the proof for $s=1$.

\smallskip

Now suppose $\calC\in\calT_\frf^{[s]}$ with $s\ge 2$. We use the induction. Assume the statement is already proven for $\calC^{(1)}\in\calT_\frf^{[s-1]}$. By~\eqref{value7} and~\eqref{value8} and the induction hypothesis, $F^{(1)}(0_+) = 1, F^{(1)}(\infty_+) = -1, F^{(1)}(0_-) = 1, F^{(1)}(\infty_-) = -1$. Defining $g$ as before to be the right-hand side of~\eqref{strCar2} multiplied by $z$, we get
\begin{equation*}
 g(0_+) = g(0_-)= \frac{1}{\alpha_0}; \quad g(\infty_+) =g(\infty_-) = \bar\alpha_0.
\end{equation*}
As above, this implies that $\deg F = \deg F^{(1)}+2$ and that $F$ satisfies~\eqref{value7} and~\eqref{value8}. Then Lemma~\ref{lemEqualDegree} shows that $\deg F = \Deg a$. Finally, let us reuse~\eqref{ffsharp}: by the induction hypothesis, $F^{(1)}(z) - F^{(1)}{}^\sharp(z)\sim z^{s-1}$, $A(z)+B(z) F^{(1)}(z) \to 2$, $A(z)+B(z) F^{(1)}{}^\sharp(z) \to 2$  as $z\to 0_+$, which proves $F(z) - F^\sharp(z)\sim z^{s}$, that is, $F$ has $z^s$ in front of $\sqrt{R_\frf(z)}$ in~\eqref{simpleF}.

%
%

\medskip
$(\calM_\frf^{s\ge 0})\Rightarrow  (\calT_\frf^{s\ge0})$
Let the Carath\'{e}odory function $F$ of some $\calC$ satisfy the conditions in $(\calM_\frf^{s\ge 0})$ for some $s\ge 0$.

If $s=0$, then $F$ is a minimal Carath\'{e}odory function (Def.~\ref{defMinimal}), so $\calC\in\calT_\frf^{[0]}$.

Suppose $s\ge 1$. By Lemma~\ref{lemInduct}, the Carath\'{e}odory function $F^{(1)}$ of $\calC^{(1)}$ is of the form
\begin{equation}\label{eqAnother}
F^{(1)}(z) = \frac{p^{(1)}(z) + z^k \sqrt{R_\frf(z)}}{a^{(1)}(z)}
\end{equation}
for some $k\ge 0$, and satisfies (2a) and (2b) of Lemma~\ref{lemEqualDegree}.

Notice that $F(0_+) = 1$, and then~\eqref{simpleF} and~\eqref{degsEqualF} with $s\ge 1$ shows that $F(0_-) = 1$ as well. By~\eqref{symmF} we have $F(\infty_+) =F(\infty_-)= -1$.  Let $g(z) = \frac{F(z)+1}{F(z)-1}$. Just as before, this is a meromorphic function on $\calS_\frf$ with $\deg g = \deg F =l+2s$.  Let us count the poles of $\tfrac{1}{z} g(z)$ with their multiplicities. For any $z_0\in \calS_\frf\setminus\{\infty_\pm,0_\pm\}$, we have  $P(z_0,\tfrac1z g) = P(z_0,g)$. Since $g(0_+)=g(0_-) = \infty$ and $g(\infty_+)=g(\infty_-) = 0$, we get $P(\infty_+,\tfrac1z g) = P(\infty_+,g)=0$, $P(\infty_-,\tfrac1z g) = P(\infty_-,g)=0$, while $P(0_+,\tfrac1z g) = P(0_+,g)-1$, $P(0_-,\tfrac1z g) = P(0_-,g)-1$. Therefore $\deg \tfrac1z g(z) = \deg g -2 = l+2s -2$. By~\eqref{strCar}, we get $\deg F^{(1)} = l+2s-2$.

If $s=1$ then $\deg F^{(1)} = l$, which shows that $ F^{(1)}$ is a minimal Carath\'{e}odory function, i.e., $\calC^{(1)}\in\calT_\frf$.

Suppose $s\ge 2$. Let us use~\eqref{strCar}, solve for $F^{(1)}$, and then compute $F^{(1)}(z)-F^{(1)}{}^\sharp(z)$. We end up with
\begin{equation}\label{ffsharp2}
F^{(1)}(z)-F^{(1)}{}^\sharp(z)  = \frac{ 4z (1-|\alpha_0|^2) \left[ F(z) - F^\sharp(z) \right] }{(C(z)+D(z) F(z))(C(z)+D(z)F^\sharp(z))},
\end{equation}
where $C(z)=(1+\bar\alpha_0)+z(1+\alpha_0)$, $D(z)= -(1+\bar\alpha_0)+z(1+\alpha_0)$.
Note that $a^{(1)}(0) \ne 0$ since $F^{(1)}(0_+)=1$. In particular, the left-hand side of~\eqref{ffsharp2} has a zero of order $k\ge 0$ at $0_+$. On the other hand, by~\eqref{simpleF},
the numerator of the right-hand side is $\sim z^{s+1}$ as $z\to 0_+$. As for the denominator, note that $F(z) = 1 + 2\alpha_0 z +O(z^2)$ as $z\to 0_+$ ($F'(0_+)=2\alpha_0$ follows by taking the limit $z\to 0_+$ in~\eqref{strCar}, applying L'H\^{o}pital's rule, and using $F^{(1)}(0_+)=1$). Therefore
\begin{equation*}\label{tempo1}
F^\sharp(z) = F(z) - \tfrac{2z^s\sqrt{R_\frf(z)}}{a(z)} = 1 + 2\alpha_0 z +O(z^2), \quad z\to 0_+
\end{equation*}
since $s\ge 2$. Then it is easy to check that
$$
(C(z)+D(z) F(z))(C(z)+D(z)F^\sharp(z)) =  4(1+\alpha_0)^2 \left(1-\alpha_0 \frac{1+\bar\alpha_0}{1+\alpha_0}\right)^2 z^2 + O(z^3)
$$
as $z\to 0_+$. Note that the coefficient in front of $z^2$ is never $0$. This means that the right-hand side of~\eqref{ffsharp2} has a zero of order $s-1$ at $0_+$. We proved that $k$ in~\eqref{eqAnother} is $s-1$. This shows that $F^{(1)}$ is of the form~\eqref{simpleF} with $k=s-1$, and we just need to justify~\eqref{degsEqualF} for $F^{(1)}$ in order to be able to apply induction. But since $k=s-1 \ge 1$, we get $F^{(1)}(0_-) = F^{(1)}(0_+) = 1$, which implies $F^{(1)}(\infty_-) = F^{(1)}(\infty_+) = -1$. Therefore part (2c) of Lemma~\ref{lemEqualDegree} holds giving us~\eqref{degsEqualF}.
\end{proof}


\section{Spectral theorem}\label{sSpectral}
As we are about to see, locations of the eigenvalues of Jacobi/CMV operators are required to satisfy a certain property with respect to the locations of the anti-bound states. Loosely speaking, every even-numbered real singularity (when counted starting from any of the edges of $\fre$ or $\frf$ in the direction of the gap) cannot be an eigenvalue and therefore must be an anti-bound state. For a lack of a better term we will call it the ``oddly interlacing'' property. Note that in particular it implies (but is stronger than) the following statement: between any two consecutive eigenvalues (which are located in the same gap) there is an odd number of anti-bound states (counted according to their multiplicities).

Let us adopt the conventions $\beta_0:=-\infty$, $\alpha_{l+2}:=+\infty$, $\theta_{2l+1}:=\theta_1+2\pi$.

\begin{definition}\label{defOI_light}
$($Oddly interlacing property$)$
\hspace*{\fill}\\
\indent {\normalfont $\rl$}
Let $\fre$ be a finite gap set~\eqref{e} on $\bbR$.
Suppose we are given two sets $($repeated according to their multiplicities$)$ of real points: $\{e_j\}_{j=1}^N$ and $\{r_j\}_{j=1}^K$ $(N,K<\infty)$. 
We will say that
$$
\{e_j\}_{j=1}^N \mbox{ oddly interlace with } \{r_j\}_{j=1}^K \mbox{ on } \bbR
$$
if
\begin{itemize}
\item $\{e_j\}_{j=1}^N \cap \{r_j\}_{j=1}^K = \varnothing$;
\item For any $k$, $0\le k \le l+1$, let $[\beta_k,\alpha_{k+1}]\cap \left( \{e_j\}_{j=1}^N \cup \{r_j\}_{j=1}^K \right) =:\{x_j\}_{j=1}^M$ $($with multiplicities preserved$)$, where
    \begin{equation}\label{interlacing_light}
    \beta_k\le x_1\le x_2 \le x_3 \le \ldots\le x_{M} \le \alpha_{k+1}.
    \end{equation}
     Then $\{x_2,x_4,\ldots\} \cap \{e_j\}_{j=1}^N = \varnothing$ and $\{x_{M-1},x_{M-3},\ldots\} \cap \{e_j\}_{j=1}^N = \varnothing.$\footnote{\label{redundant}The second condition will turn out to be redundant since $M$ will always end up being finite and odd here. We keep it this way to agree with a more general case~\cite{K_spectral}.}
\end{itemize}

\smallskip{\normalfont $\uc$}
Let $\frf$ be a finite gap set~\eqref{e2} on $\partial\bbD$.
Suppose we are given two sets $($repeated according to their multiplicities$)$ of unimodular points: $\{e_j\}_{j=1}^N$ and $\{r_j\}_{j=1}^K$ with $N,K<\infty$ and $e_j\in\partial\bbD$, $r_j\in\partial\bbD$ for all $j$. 
We will say that
$$
\{e_j\}_{j=1}^N \mbox{ oddly interlace with } \{r_j\}_{j=1}^K \mbox{ on } \partial\bbD
$$
if
\begin{itemize}
\item $\{e_j\}_{j=1}^N \cap \{r_j\}_{j=1}^K = \varnothing$;
\item For any $k$, $1\le k \le l$, let $\{e^{i\theta}: \theta_{2k}\le \theta \le \theta_{2k+1}\} \cap \left(\{e_j\}_{j=1}^N \cup \{r_j\}_{j=1}^K \right) =:\{e^{ix_j}\}_{j=1}^M$ $($with multiplicities preserved$)$, where
    \begin{equation}\label{interlacing_lightF}
    \theta_{2k}\le x_1\le x_2 \le x_3 \le \ldots\le x_{M} \le \theta_{2k+1}.
    \end{equation}
     Then $\{e^{i x_2},e^{i x_4},\ldots\} \cap \{e_j\}_{j=1}^N = \varnothing$ and $\{e^{i x_{M-1}},e^{i x_{M-3}},\ldots\} \cap \{e_j\}_{j=1}^N = \varnothing.$\textsuperscript{$\ref{redundant}$}
\end{itemize}
\end{definition}
\begin{remarks}
1. If $N=0$ (no eigenvalues), then this property trivially holds for any configuration of $\{r_j\}$.

2. If we think of $e_j$'s as eigenvalues, $r_j$'s as anti-bound states, 
then this property states that every even-numbered real singularity (when counted starting from any of the edges in the direction of the gap), must be an anti-bound state.
\end{remarks}

Now we can state the characterization of the spectral measures. Let us define
\begin{align*}
\sg_\fre(x) & = \begin{cases}
(-1)^{l+1-k} & \mbox{ if } x\in (\alpha_k,\beta_k), \\
0 & \mbox{ otherwise},
\end{cases} \\
\sg_\frf(\theta) & = \begin{cases}
(-1)^{k-1} & \mbox{ if } \theta\in (\theta_{2k-1},\theta_{2k}), \\
0 & \mbox{ otherwise},
\end{cases}
\end{align*}
the functions that change sign from one band to another.

\begin{theorem}\label{thmS}
\hspace*{\fill}\\
\indent {\normalfont $\rl$}
The following are equivalent:
\begin{itemize}
\item[$(\calT_\fre^{s\ge 0})$] Jacobi matrix $\calJ[a_n,b_n]_{n=1}^\infty$ belongs to $\calT_\fre^{[s]}$ $(s\ge 0)$ $($see Def.~\ref{defFiniteRange}$)$. 
\item[$(S_\fre^{s\ge 0})$] The spectral measure $\mu$ of $\calJ$ is of the form
\begin{equation}\label{meas2}
d\mu(x) = \frac{\sqrt{|R_\fre(x)|}}{|d(x)|} 1_{x\in\fre} dx + \sum_{j=1}^N w_j \delta_{E_j},
\end{equation}
where
\begin{itemize}
\item[$(S_a)$] $d(z)$ is a real polynomial of degree $\Deg d = l + s$ which satisfies\footnote{Up to a normalization, this condition is equivalent to saying that all the zeros of $d$ are either real or come in complex-conjugate pairs, and that there is an odd number of zeros in each gap (in particular $\deg d \ge l$). See Theorem~\ref{thmR} below.}
     $\sgn \,d(x) = \sg_\fre(x)$
    on $\operatorname{Int}(\fre)$; 
  \item[$(S_b)$]  $N<\infty$ and $E_j\in \bbR\setminus\fre$. Each $E_j$ is a simple zero of $d(z)$. Moreover, 
  $\{E_j\}_{j=1}^N$  oddly interlace $($Def.~\ref{defOI_light}$)$ with
  \begin{equation}\label{spResonances}
  \{R_j\}_{j=1}^{K}:=\left\{\mbox{zeros of }d(z)\mbox{ in }\bbR\setminus\{E_j\}_{j=1}^N\right\},
  \end{equation}
  repeated according to their multiplicities;
  \item[$(S_c)$] For each  $1\le j \le N$,
  \begin{equation}\label{weights2}
  w_j= 2\pi \frac{\sqrt{R_\fre(E_j)}}{\left| d'(E_j) \right| }.
  \end{equation}
\end{itemize}
\end{itemize}


\smallskip{\normalfont $\uc$}
The following are equivalent:
\begin{itemize}
\item[$(\calT_\frf^{s\ge 0})$] CMV matrix $\calC[\alpha_n]_{n=0}^\infty$ belongs to $\calT_\frf^{[s]}$ $(s\ge 0)$ $($see Def.~\ref{defFiniteRange}$)$.
\item[$(S_\frf^{s\ge 0})$] The spectral measure $\mu$ of $\calC$ is of the form
\begin{equation}\label{meas2F}
d\mu(\theta) = \frac{\sqrt{|R_\frf(\theta)|}}{|d(e^{i \theta})|} 1_{\theta\in\frf} d\theta + \sum_{j=1}^N w_j \delta_{E_j},
\end{equation}
where
\begin{itemize}
\item[$(S_a)$] $d(z)$ is a polynomial of degree $l+2s$, and\footnote{By $\operatorname{Int}(\frf)$ here we mean $\frf\setminus\cup_{j=1}^{2l}\{\theta_j\}$} on $\operatorname{Int}(\frf)$ it satisfies\footnote{\label{footnote}Since $\Deg d =l+2s$, up to a normalization, this condition is equivalent to saying that all the zeros of $d$ are either unimodular or come in symmetric (with respect to $\partial\bbD$) pairs, and that there is an odd number of zeros in each gap. See Theorem~\ref{thmR} below.}\textsuperscript{,}\footnote{If $l$ is odd, then this condition has $z^{-1/2}$. One can just choose any branch of the square root with a branch cut that goes through the last gap $(\theta_{2l}-2\pi,\theta_1)$ (see the discussion in the end of Subsection~\ref{ssMeromorphic}). Alternatively, the comment in~\textsuperscript{\ref{footnote}} is still valid for $l$ odd.}
    \begin{equation}\label{messySign}
    \sgn \left[ e^{-is\theta-il\theta/2}\,d(e^{i \theta}) \right] = \sg_\frf(\theta);
    \end{equation}
  \item[$(S_b)$]  $N<\infty$ and $E_j\in \partial\bbD\setminus\frf$. Each $E_j$ is a simple zero of $d(z)$. Moreover, 
  $\{E_j\}_{j=1}^N$  oddly interlace $($Def.~\ref{defOI_light}$)$ with
  \begin{equation}\label{spResonancesF}
  \{R_j\}_{j=1}^{K}:=\left\{\mbox{zeros of }d(z)\mbox{ in }\partial\bbD\setminus\{E_j\}_{j=1}^N\right\},
  \end{equation}
  repeated according to their multiplicities;
  \item[$(S_c)$] For each  $1\le j \le N$,
  \begin{equation}\label{weights2F}
  w_j= 2\pi \frac{\sqrt{ |R_\frf(E_j) | }}{\left| d'(E_j) \right| }.
  \end{equation}
\end{itemize}
\end{itemize}
\end{theorem}
\begin{remarks}
1. We stress that $(S_b)$ is a statement about which points are allowed to be eigenvalues. There is no implicit restriction on $d(z)$ here, and any function $d(z)$ that satisfies ($S_a$) (up to a multiplicative normalization constant) can occur in~\eqref{meas2}/\eqref{meas2F}.

2. Similarly, $(S_c)$ is a statement about the eigenweights only (again, up to an inconsequential normalization). Indeed, note that each $w_j$ in~\eqref{weights2}/\eqref{weights2F} is positive, so there is no implicit positivity restriction here either.



3. As is clear from $(2a)$ of Lemma~\ref{lemEqualDegree}, resonances cannot occur at the points which are eigenvalues. Therefore~\eqref{continuation}/\eqref{continuationF} show that in this case resonances occur precisely at those zeros of $d(z)$ that are not eigenvalues. 
This explains that the $\{R_j\}_{j=1}^K$ in $(S_b)$ are precisely the anti-bound states 
of the operator.
\end{remarks}

\begin{proof}
\rl

$(\calT_\fre^{s\ge 0})\Rightarrow(S_\fre^{s\ge 0})$
Let $\calJ\in\calT_\fre^{[s]}$, $s\ge 0$. Then its $m$-function is of the form
\begin{equation}\label{spEq0}
m(z)=\frac{p(z)+\sqrt{R_\fre(x)}}{a(x)},
\end{equation}
and satisfies Def.~\ref{defMinimal} or $(\calM_\fre^{s\ge 1})$ of Theorem~\ref{thmM}. By the computation~\eqref{densityM} and the fact that $i\sqrt{R_\fre(x)}$ changes sign from one band of $\fre$ to another, we obtain that $a$ is real, $\sgn \, a(x) = \sg_\fre(x)$ on $\fre$, and the a.c. density of $\mu$ is therefore
\begin{equation}\label{spEq1}
\frac{d\mu(x)}{dx}= \frac{\sqrt{|R_\fre(x)|}}{\pi |a(x)|}
\end{equation}
on $\fre$. This proves $(S_a)$ if one takes $d(x)=\pi a(x)$.

By the Herglotz representation, Lemma~\ref{herglotz}, each of the eigenvalues $E_j$ of $\calJ$ must be a pole of~\eqref{spEq0}, and therefore a zero of the polynomial $d(z)$. Moreover, by (2b) of Lemma~\ref{lemEqualDegree}, $m$ has at most $\tfrac{1}{(z-z_0)^{1/2}}$ singularity when $z_0\in\cup_{j=1}^{l+1}\{\alpha_j,\beta_j\}$, which means the endpoints cannot be eigenvalues by applying Herglotz representation again. Also, $E_j\notin\operatorname{Int}(\fre)$  since~\eqref{spEq1} must be integrable. Note that $(E_j)_+$ is always a \textit{simple} pole of $m$ by~\eqref{m}. Therefore by~\eqref{continuation}, if $E_j$ is a zero of $a$ of order higher than $1$, then $m$ would also have a pole at $(E_j)_-$ which is impossible by $(2a)$ of Lemma~\ref{lemEqualDegree}. Therefore each $E_j$ is a simple zero of $d$.

Now we need to show that $\{E_j\}_{j=1}^N$ oddly interlace with $\{R_j\}_{j=1}^K$ defined by~\eqref{spResonances}. Notice that $m((R_j)_+)$ is finite, which implies that each $(R_j)_-$ is indeed a pole of $m$ by~\eqref{continuation}, that is, $\{R_j\}_{j=1}^K$ are the anti-bound states of $\calJ$. Fix some $0\le k \le l+1$ and order the singularities on $[\beta_k,\alpha_{k+1}]$ as in~\eqref{interlacing_light}. $\beta_k$ is at most a first order pole of $m$, so the Taylor series of $m$ at $\beta_k$ is of the form
\begin{equation*}\label{taylor}
m(z)=\sum_{j=-1}^\infty k_j (z-\beta_k)^{j/2},
\end{equation*}
which implies
\begin{equation}\label{temp1}
m(z)-m^\sharp(z) =2k_{-1}(z-\beta_k)^{-1/2}+ 2k_1 (z-\beta_k)^{1/2} + O((z-\beta_k)^{3/2}).
\end{equation}
By~\eqref{densityM}, the left-hand side of~\eqref{temp1} belongs to $i\bbR_+$ on $\pi^{-1}((\alpha_k,\beta_k))\cap \calS_{\fre,+}$. Choose for definiteness $(z-\beta_k)^{1/2}$ to be positive for $z\in\calS_{\fre,+}$, $\pi(z)>\beta_k$. Then in order for the right-hand side of~\eqref{temp1} to be in $i\bbR_+$ on $\pi^{-1}((\alpha_k,\beta_k))\cap \calS_{\fre,+}$, we  need
 either $k_{-1}<0$, or $k_{-1}=0$ and $k_1>0$ (note that it is not possible to have $k_{-1}=k_1=0$ because of~\eqref{continuation}).

If $k_{-1}<0$, then $m-m^\sharp$ is negative on $\calS_{\fre,+}$ to the right of $\beta_k$. Since in this case $\beta_k$ is a first order pole, $x_1=\beta_k$. Note that $m-m^\sharp$ never vanishes on $\pi^{-1}(\bbC\setminus\cup_{j=1}^{l+1}\{\alpha_j,\beta_j\})$ by~\eqref{continuation}, and thus
\begin{equation}\label{residue}
\lim_{z\to (x_2-0)_+} m(z)-m^\sharp(z)=-\infty
\end{equation}
if $M\ge 2$.
Now, if $x_2$ were an eigenvalue, then $\lim_{z\to (x_2-0)_+} m(z)=+\infty$ (by~\eqref{m}) and $\lim_{z\to (x_2-0)_+} m^\sharp(z)$ is finite by (2a) of Lemma~\ref{lemEqualDegree}, which would contradict to \eqref{residue}. This implies that $x_2$ is an anti-bound state.

Let us now consider the case $k_{-1}=0$ and $k_1>0$. $m-m^\sharp$ is positive on $\calS_{\fre,+}$ to the right of $\beta_k$. Since $m-m^\sharp$ cannot be equal to zero, we obtain $$\lim_{z\to (x_1-0)_+} m(z)-m^\sharp(z)=+\infty.$$
If $x_1=x_2$ then it is a resonance, since $m$ is Herglotz on $\calS_{\fre,+}$ and therefore cannot have second order poles there. If $x_1\ne x_2$, then $m-m^\sharp\ne 0$ on $\pi^{-1}((\beta_k,\alpha_{k+1}))$ gives
\begin{align*}
&\lim_{z\to (x_1+0)_+} m(z)-m^\sharp(z)=-\infty,\\
&\lim_{z\to (x_2-0)_+} m(z)-m^\sharp(z)=-\infty,
\end{align*}
which implies that $x_2$ is an anti-bound state by the same arguments as above. Checking the signs of $m-m^\sharp$ further, one sees that~\eqref{residue} holds at any $x_j$ with even $j$, which means they are anti-bound states.

The arguments for $\{x_{M-1},x_{M-3},\ldots\}$ are analogous if one examines the signs into the gap starting from the edge $\alpha_{k+1}$. This proves $(S_b)$.

To prove $(S_c)$ let us put $z=(E_j)_+$ in~\eqref{continuation}, and take residues of both sides. The residue of $m$ is $-w_j$ by~\eqref{m}, and since $(E_j)_-\in\calS_-$ cannot be a pole of $m$ by (2a) of Lemma~\ref{lemEqualDegree}, the residue of $m^\sharp$ is zero. Therefore
\begin{equation}\label{weights_complex}
w_j = -2 \res_{z=(E_j)_+} \frac{\sqrt{R_\fre (z)}}{a(z)} = -2\pi \frac{\sqrt{R_\fre ((E_j)_+)}}{d'(E_j)}.
\end{equation}
Finally we note that the latter expression is automatically positive given $(S_b)$ and therefore is equal to~\eqref{weights2}.
Suppose  that $E_j\in(\beta_{l+1},+\infty)$.  By $(S_a)$, $d(z)$ is positive on $(\alpha_{l+1},\beta_{l+1})$, and by the oddly interlacing property, there is an even number of  zeros of $d(z)$ (counting with multiplicities) on the interval $[\beta_{l+1},E_j)$ . Thus $d'(E_j)<0$, and since $\sqrt{R_\fre((E_j)_+)}>0$, we conclude that the right-hand side of~\eqref{weights_complex} is positive. The arguments for $E_j$'s on any of the gaps or on $(-\infty,\alpha_1)$ are similar if one uses the sign condition on $d(z)$ from $(S_a)$, the sign changes of $\sqrt{R_\fre(z)}$, and the oddly interlacing property.

\smallskip

$(S_\fre^{s\ge 0})\Rightarrow(\calT_\fre^{s\ge 0})$
Note that $(S_a)$ requires $d$ to have at least $l$ zeros. The case $\Deg d = l$ corresponds to $\calJ\in\calT_\fre$ and is well-known. Suppose $\mu$ satisfies $(S_\fre^{s\ge 0})$ with $\Deg d \ge l+1$, and let $m$ be the $m$-function~\eqref{m}. Define $\tilde{m}(z) = m(z)-\frac{\pi \sqrt{R_\fre(z)}}{d(z)}$ on $\bbC\setminus\fre$. By the Herglotz representation, Lemma~\ref{herglotz}, $\imag \tilde{m}(x+i\veps) = \imag \tilde{m}(x-i\veps) = 0$ and $\real\tilde{m}(x+i\veps) = \real \tilde{m}(x-i\veps)$ for $x\in\fre$. This shows that $\tilde{m}$ has a meromorphic continuation to $\bbC$ (there is a small issue at the endpoints $z_0\in\cup_{j=1}^{l+1}\{\alpha_j,\beta_j\}$, which can be resolved by directly showing that $|(z-z_0)\tilde{m}(z)|$ is bounded around $z_0$ and therefore $z_0$ cannot be an essential singularity). Note that both $m$ and $\frac{\sqrt{R_\fre(z)}}{d(z)}$ have limits (possibly infinite) as $z\to\infty$. This implies that $\tilde{m}$ is a meromorphic function on $\bbC\cup\{\infty\}$, and therefore must be a rational function. This shows that $m$ has a meromorphic continuation to $\calS_\fre$ and is of the form
$$
m(z) = \frac{p(z) + p_2(z)\sqrt{R_\fre(z)}}{p_2(z)a(z) }
$$
for some polynomials $p, p_2$ with no common zeros. Suppose that $p_2(z_0)=0, p(z_0)\ne 0$. If $a(z_0)=0$ then $m$ has a pole of order $\ge2$ at $(z_0)_+$ which implies that $z_0$ is not an eigenvalue by the Herglotz representation, Lemma~\ref{herglotz}. If $a(z_0) \ne 0$ then again $z_0$  cannot be an eigenvalue by $(S_b)$. But then $m$ must be regular at $(z_0)_+$ which contradicts to $p(z_0)\ne 0$. We proved that $p_2$ must be a constant which may be divided out to produce~\eqref{spEq0}. We claim that $m$ satisfies Lemma~\ref{lemEqualDegree}. Indeed, $(2a)$ follows by taking the residues of~\eqref{continuation} and using~\eqref{weights_complex} (note that~\eqref{weights_complex} is equal to~\eqref{weights2} by $(S_b)$ as we showed above). $(2b)$ follows since second order pole of $a(z)$ at an endpoint of $\fre$ would make $\mu$ non-integrable. Finally, $(2c)$ follows from $m(z) = -\tfrac{1}{z}$, $z\to\infty_+$,~\eqref{continuation}, and $\Deg R_\fre = 2(l+1)$, $\Deg d \ge l+1$. Theorem~\ref{thmM} finishes the proof.

\smallskip{\normalfont $\uc$}


$(\calT_\frf^{s\ge 0})\Rightarrow(S_\frf^{s\ge 0})$
Let $\calC\in\calT_\frf^{[s]}$, $s\ge 0$. Then its Carath\'{e}odory function $F$ satisfies $(\calM_\fre^{s\ge 0})$ of Theorem~\ref{thmM}. By the computation~\eqref{densityF} and the fact that $e^{-il\theta/2}\sqrt{R_\frf(e^{i\theta})}$ is purely imaginary and changes sign from one band of $\frf$ to another, we obtain that $e^{-is\theta-il\theta/2} a(e^{i\theta})$ is purely imaginary, $\sgn \left[ \tfrac{1}{i} e^{-is\theta-il\theta/2} a(e^{i\theta})  \right] = \sg_\frf(\theta)$ on $\frf$, and the a.c. density of $\mu$ is therefore
\begin{equation}\label{spEq1F}
\frac{d\mu(\theta)}{d\theta}= \frac{\sqrt{|R_\frf(\theta)|}}{2 \pi |a(e^{i\theta})|}
\end{equation}
on $\frf$. This proves $(S_a)$ if one takes $d(z)=-2\pi i a(z)$.

By the Herglotz representation, Lemma~\ref{herglotz}, each of the point masses $E_j$ of $\mu$ must be a pole of $F$, and therefore a zero of the polynomial $d(z)$. Moreover, by (2b) of Lemma~\ref{lemEqualDegree}, $F$ has at most $\tfrac{1}{(z-z_0)^{1/2}}$ singularity when $z_0\in\cup_{j=1}^{2l}\{e^{i\theta_j}\}$, which means the endpoints cannot be eigenvalues by applying Herglotz representation again. Also, $E_j\notin\operatorname{Int}(\frf)$  since~\eqref{spEq1F} must be integrable. Note that $(E_j)_+$ is always a \textit{simple} pole of $F$ by~\eqref{F}. Therefore by
\begin{equation}\label{continuationF}
F(z) = F^\sharp(z) + \frac{2 z^s \sqrt{R_\frf(z)}}{a(z)},
\end{equation}
if $E_j$ is a zero of $a$ of order higher than $1$, then $F$ would also have a pole at $(E_j)_-$ which is impossible by $(2a)$ of Lemma~\ref{lemEqualDegree}. Therefore each $E_j$ is a simple zero of $d$.

Now we need to show that $\{E_j\}_{j=1}^N$ oddly interlace with $\{R_j\}_{j=1}^K$ defined by~\eqref{spResonancesF}. Notice that $F((R_j)_+)$ is finite, which implies that each $(R_j)_-$ is indeed a pole of $F$ by~\eqref{continuationF}, that is, $\{R_j\}_{j=1}^K$ are the anti-bound states of $\calC$. Fix some $1\le k \le l$ and order the singularities on the gap $G_j$ (see~\eqref{gapF}) as in~\eqref{interlacing_lightF}. By~\eqref{symmF} (and its analytic continuation to $\calS_{\frf,-}$), $F((e^{i\theta})_+)-F^\sharp((e^{i\theta})_+)$ is purely imaginary for $\theta\in[\theta_{2k},\theta_{2k+1}]$. Consider two cases:

(i) $(e^{i\theta_{2k}})_+$ is not a pole of $F$. Note that by definition (see Subsection~\ref{ssMeromorphic}), $e^{-il\theta/2}\sqrt{R_\frf((e^{i\theta})_+)}$ belongs to $(-1)^{k-1}i\bbR_+$ on $[\theta_{2k-1},\theta_{2k}]$ and to $(-1)^{k-1} \bbR_+$ on $[\theta_{2k},\theta_{2k+1}]$. As we just established, $e^{-is\theta-il\theta/2} a(e^{i\theta})$ is in $(-1)^{k-1}i\bbR_+$ for $\theta\in[\theta_{2k-1},\theta_{2k}]$, as well as for $\theta$'s slightly to the right of $\theta_{2k}$ since $e^{i\theta_{2k}}$ is not a zero of $a$. Therefore by~\eqref{continuationF}, $F((e^{i\theta})_+)-F^\sharp((e^{i\theta})_+)$ is in $\tfrac{(-1)^{k-1}}{(-1)^{k-1}i} \bbR_+ = -i \bbR_+$ for $\theta = \theta_{2k} + \veps$, $0<\veps \ll 1$. Note that $F-F^\sharp$ never vanishes inside the gap, but goes to infinity at $(e^{i x_2})_+$ while being purely imaginary. We can conclude that $F((e^{i\theta})_+)-F^\sharp((e^{i\theta})_+) \to -i\infty$ when $\theta \to x_1-0$, so $\to +i\infty$ when $\theta\to x_1+0$, and therefore $\to +i\infty$ when $\theta\to x_2-0$ (we assume $x_1\ne x_2$, which can be treated similarly). If $e^{i x_2}$ were an eigenvalue, then by~\eqref{cara2} we would have
\begin{align}
\mu(\{x_2\}) & = \lim_{r \uparrow 1} \left( \frac{1-r}{2}\right) F((r e^{i x_2})_+) \\
\label{ei1} & = -\tfrac{1}{2} e^{-i x_2} \res_{z=(e^{i x_2})_+} F(z) \\
\label{ei2} &= -i  \lim_{\theta\to x_2} \sin\left(\tfrac{\theta-x_2}{2}\right) F((e^{i\theta})_+) \\
\label{ei3} & = -i  \lim_{\theta\to x_2-0} \sin\left(\tfrac{\theta-x_2}{2}\right) \left[F((e^{i\theta})_+)-F^\sharp((e^{i\theta})_+)\right],
\end{align}
where~\eqref{ei1} comes from writing the definition of the residue and taking the limit $z\to (e^{ix_2})_+$ along $(re^{ix_2})_+$, $r\uparrow1$; \eqref{ei2} comes from taking the same limit along $(e^{i\theta})_+$, $\theta\to x_2$; and~\eqref{ei3} follows from regularity of $F$ at $(e^{ix_2})_-$ (by $(2a)$ of Lemma~\ref{lemEqualDegree}). But then using $F((e^{i\theta})_+)-F^\sharp((e^{i\theta})_+) \to +i\infty$ when $\theta\to x_2-0$, we obtain $\mu(\{x_2\})\le 0$, a contradiction. Therefore $e^{ix_2}$ must be a resonance.

(ii) If $(e^{i\theta_{2k}})_+$ is a pole of $F$ (then $x_1=\theta_{2k}$, of course), then $e^{-is\theta-il\theta/2} a(e^{i\theta})$ is in $(-1)^k i\bbR_+$ for $\theta$'s immediately to the right of $\theta_{2k}$. Then $F((e^{i\theta})_+)-F^\sharp((e^{i\theta})_+)$ is in $i \bbR_+$  to the right of $\theta_{2k}$, and therefore $F((e^{i\theta})_+)-F^\sharp((e^{i\theta})_+) \to +i\infty$ when $\theta \to x_2-0$. The rest of the arguments in (i) show that $e^{ix_2}$ is then a resonance.

That the rest of $e^{ix_{2j}}$ and $e^{ix_{M-2j}}$ are resonances can be shown in the exact same way.


To prove $(S_c)$ let us put $z=(E_j)_+$ in~\eqref{continuationF}, and take residues of both sides. The residue of $F$ is $-2 E_j w_j$ by~\eqref{ei1}, and since $(E_j)_-\in\calS_{\frf,-}$ cannot be a pole of $F$ by (2a) of Lemma~\ref{lemEqualDegree}, the residue of $F^\sharp$ is zero. Therefore
\begin{equation*}
-2 E_j w_j = 2 \res_{z=(E_j)_+} \frac{z^s\sqrt{R_\frf (z)}}{a(z)} = 4\pi \frac{ E^s_j \sqrt{R_\frf ((E_j)_+)}}{i \, d'(E_j)},
\end{equation*}
which is equivalent to
\begin{equation}\label{weights_complexF}
 w_j = - 2\pi \frac{ E^{-l/2}_j \sqrt{R_\frf ((E_j)_+)}}{i E_j^{-s-l/2+1}\, d'(E_j)}.
\end{equation}
Finally we note that the latter expression is automatically positive given $(S_b)$ and therefore is equal to~\eqref{weights2F}.
Indeed, suppose  that $E_j = e^{i x_0}$ with some $x_0\in(\theta_{2k},\theta_{2k+1})$.  By $(S_a)$, $e^{-is\theta-il\theta/2} d(e^{i\theta})$ has sign $(-1)^{k-1}$  on $(\theta_{2k-1},\theta_{2k})$, and by $(S_b)$ there is an even number of  zeros of $d(e^{i\theta})$ (counting with multiplicities) on the interval $[\theta_{2k},x_0)$ . Thus
$$
\frac{d}{d\theta} e^{-is\theta-il\theta/2} d(e^{i\theta}) \Big|_{\theta=x_0}= i E_j^{-s-l/2+1} d'(E_j)
$$
has $(-1)^k$ sign. Since $E_j^{-l/2} \sqrt{R_\frf((E_j)_+)} \in (-1)^{k-1}\bbR_+$, we conclude that the right-hand side of~\eqref{weights_complexF} is positive.

\smallskip

$(S_\frf^{s\ge 0})\Rightarrow(\calT_\frf^{s\ge 0})$
Suppose $\mu$ satisfies $(S_\frf^{s\ge 0})$, and let $F$ be its Carath\'{e}odory function~\eqref{F}. Repeating the arguments from $(S_\fre^{s\ge 0})\Rightarrow(\calT_\fre^{s\ge 0})$, we see that $F$ is of the form
$$
F(z) = \frac{p(z) + z^s \sqrt{R_\frf(z)}}{a(z) }
$$
where $a(z) = -\tfrac{1}{2\pi i} d(z)$. Then one checks that $F$ satisfies Lemma~\ref{lemEqualDegree}: $(2a)$ follows by taking the residues of~\eqref{continuationF} and using~\eqref{weights_complexF} (note that~\eqref{weights_complexF} is equal to~\eqref{weights2F} by $(S_b)$ as we showed above); $(2b)$ follows since second order pole of $a(z)$ at an endpoint of $\frf$ would make $\mu$ non-integrable; $(2c)$ follows from $F(0_+) = 1$,~\eqref{continuationF}, and $\Deg R_\fre = 2l$, $\Deg a = l+2s$. Theorem~\ref{thmM} finishes the proof.
\end{proof}

%
%

\section{Inverse resonance problem: existence and uniqueness}\label{sResonance}

We can now solve the inverse resonance problem: we give necessary and sufficient conditions for a configuration of points to be the eigenvalues and resonances of the operators from $\calT^{[s]}$, and show that such an operator is unique. Equivalently, we can characterize all the poles of the functions $m$ from $\calM_\fre^{s\ge 0}$ and $F$ from $\calM_\frf^{s \ge 0}$ (see Theorem~\ref{thmM}). 

\begin{theorem}\label{thmR}
\hspace*{\fill}\\
\indent {\normalfont $\rl$}
Let $\{R_j\}_{j=1}^{K}$ and $\{E_j\}_{j=1}^N$ $(0\le N,K< \infty)$ be two sequences of complex numbers $($possibly with multiplicities$)$. These two sequences are respectively resonances and eigenvalues of a Jacobi operator from $\calT_\fre^{[s]}$ $(s\ge 0)$
   if and only if 
\begin{itemize}
\item[$(O_1)$] $\{E_j\}_{j=1}^N$ oddly interlace with $\{R_j\}_{j=1}^{K} \cap \bbR$ on $\bbR$ $($see Def.~\ref{defOI_light}$)$;\footnote{\label{reminder1}We remind that $(O_1)$ includes $\{R_j\}_{j=1}^K\cap\{E_j\}_{j=1}^N = \varnothing$ as part of the Definition~\ref{defOI_light}.}
\item[$(O_2)$] Each gap $[\beta_k,\alpha_{k+1}]$ contains an odd number of points from $\{E_j\}_{j=1}^N \cup \{R_j\}_{j=1}^{K}$ $($counting with multiplicities$)$;
\item[$(O_3)$] $E_j\in\bbR\setminus\fre$ for every $j$; each $E_j$ is of multiplicity $1$;
\item[$(O_4)$] $R_j\in\bbC\setminus\Int(\fre)$ and they are real or come in complex conjugate pairs $($counting multiplicities$)$; if $R_{j}\in\cup_{j=1}^{l+1}\{\alpha_j,\beta_j\}$, then the multiplicity of $R_{j}$ is $1$;
\item[$(O_5)$] $K+N = l+s$.
\end{itemize}
 Such a Jacobi operator $\calJ$ is unique. 

\smallskip{\normalfont $\uc{}$}
Let $\{R_j\}_{j=1}^{K}$ and $\{E_j\}_{j=1}^N$ $(0\le N,K< \infty)$ be two sequences of complex numbers $($possibly with multiplicities$)$. These two sequences are respectively resonances and eigenvalues of a CMV operator from $\calT_\frf^{[s]}$
   if and only if 
\begin{itemize}
\item[$(O_1)$] $\{E_j\}_{j=1}^N$ oddly interlace with $\{R_j\}_{j=1}^{K} \cap\partial\bbD$ on $\partial\bbD$ $($see Def.~\ref{defOI_light}$)$;\textsuperscript{\ref{reminder1}}
\item[$(O_2)$] Each gap $[\beta_k,\alpha_{k+1}]$ contains an odd number of points from $\{E_j\}_{j=1}^N \cup \{R_j\}_{j=1}^{K}$ $($counting with multiplicities$)$;
\item[$(O_3)$] $E_j\in\partial\bbD\setminus\frf$ for every $j$; each $E_j$ is of multiplicity $1$;
\item[$(O_4)$] $R_j\in\bbC\setminus\{0\} \setminus\Int(\frf)$ and they are unimodular or come in symmetric (with respect to $\partial\bbD$) pairs $($counting multiplicities$)$; if $R_{j}\in\cup_{j=1}^{2l}\{e^{i\theta_j}\}$, then the multiplicity of $R_{j}$ is $1$;
\item[$(O_5)$] $K+N = l+2s$.
\end{itemize}
 Such a CMV operator $\calC$ is unique.
\end{theorem}
\begin{remark}
In particular for $\uc{}$ the parity of $l$ and of the total number of singularities $K+N$ must coincide.
\end{remark}




\begin{proof}

The arguments for $\rl{}$ and $\uc{}$ are almost identical here. Let us show the $\uc{}$ case only.

\smallskip{\normalfont $\uc{}$}

Let us first show the necessity. Theorem~\ref{thmS} contains $(O_1)$ and $(O_3)$ in $(S_b)$. $(O_2)$ follows from the sign-alternating property of $d(e^{i\theta})$, see $(S_a)$. That $R_j$ are unimodular or come in symmetric (with respect to $\partial\bbD$) pairs follows from the sign condition in $(S_b)$: indeed, $d(z)$ and $\overline{d(\bar{z}^{-1})}$ coincide on $\partial\bbD$ (on $\partial\bbD$ away from the cut, if $l$ is odd), which implies that zeros are symmetric. The rest of $(O_4)$ is a consequence of integrability of $\tfrac{d\mu}{dx}$ on $\fre$. $(O_5)$ is clear from the degree condition of $(S_a)$.

To show sufficiency, given $\{R_j\}_{j=1}^{K}$ and $\{E_j\}_{j=1}^N$, let
$$d(z)=A \prod_{j=1}^{K} R_j^{-1/2} (z-R_j) \prod_{j=1}^{N} E_j^{-1/2} (z-E_j),$$ where $A$ is a real constant to be determined momentarily. Note that $e^{-is\theta-il\theta/2}\,d(e^{i \theta})$ is real on $\partial\bbD$ by the analogue of~\eqref{reality}. Now choose the sign of $A$ so that~\eqref{messySign} holds on the first band $(\theta_1,\theta_2)$. Using $(O_2)$, we can see that~\eqref{messySign} holds on each of the subsequent bands of $\frf$ too. Define $w_j>0$ by \eqref{weights2F} for each $1\le j\le N$. 
Finally, the absolute value of $A$ can be chosen so that the total mass of $\mu$ is $1$.

Uniqueness follows from the fact that each step of the measure reconstruction was uniquely determined by the spectral characterization of Theorem~\ref{thmS}.
\end{proof}

\section{$m$-functions as solutions to an interpolation problem}\label{sInterpolation}

In the previous section we showed how one can recover the spectral measure from the resonances and eigenvalues. The $m$-function is then, of course, just~\eqref{m}. Let us conclude this paper by showing explicitly and constructively how one can recover $m$ from  $\{R_j\}_{j=1}^K$ and $\{E_j\}_{j=1}^N$ without doing the integration in~\eqref{m}. The arguments for Carath\'{e}odory functions can be done in the analogous way and will be skipped.

For simplicity let us assume that $R_j\ne R_k$ for $j\ne k$, i.e., each resonance has multiplicity 1. We will discuss the changes necessary for the general case in the end of the section.

%

From the discussion above, we know that
$$
m(z) = \frac{p(z) + \sqrt{R_\fre(z)}}{a(z)},
$$
where $a(z) = A \prod_{j=1}^{K} (z-R_j) \prod_{j=1}^{N} (z-E_j)$, where the sign of $A\in\bbR$ is chosen so that $a(z)$ is positive on $(\alpha_{l+1},\beta_{l+1})$ and the absolute value will be chosen later to normalize $\lim_{z\to\infty_+} zm(z) = -1$. The polynomial $p(z)$ can be recovered from the condition $(2)$ of Lemma~\ref{lemEqualDegree}. Indeed, we claim it must satisfy
\begin{equation}\label{system}
\begin{cases}
\left(p(z)-\sqrt{R_\fre(z)}\right)\big|_{z=(E_j)_+} = 0, \quad j=1,\ldots,N; \\
 \left(p(z) +  \sqrt{R_\fre(z)} \right)\big|_{z=(R_j)_+} =0, \quad j=1,\ldots,M; \\
\frac{p(z)+\sqrt{R_\fre(z)}}{z^{N+M}}\to 0, \quad z\to\infty_+.
\end{cases}
\end{equation}
Indeed, the first two equations come from $(2a)$ and $(2b)$ of Lemma~\ref{lemEqualDegree} (note that if $R_j$ is an endpoint of $\fre$, then $p(R_j)=0$ by the arguments in the proof of Lemma~\ref{lemEqualDegree}), and the last condition of~\eqref{system} is a consequence of $m(z) \to 0$ as $z\to\infty_+$. Now let us show that this system determines $p(z)$ uniquely.

Note that the first two lines of~\eqref{system} constitute $N+M$ linear equations  with respect to the unknown coefficients of the polynomial
$$p(z)=\sum_{k=0}^L c_k z^k.$$

Consider the following cases.

If $K+N = l = \Deg a$ then (recall $\Deg R_\fre = 2(l+1)$) the last condition of~\eqref{system} requires $L=l+1$, and determines the coefficients $c_{l+1}$ and $c_{l}$. Therefore we are left with $l$ unknown coefficients $c_{l-1},\ldots,c_0$. Note that this coincides with the number of the linear equations in~\eqref{system}.

If $K+N=l+1 = \Deg a$ then the last condition of~\eqref{system} requires $L=l+1$ and determines only the coefficient $c_{l+1}=-1$. Therefore we are left with $l+1$ unknown coefficients $c_{l},\ldots,c_0$. Note that again, this coincides with the number of the linear equations in~\eqref{system}.

Finally, if $K+N = \Deg a \ge l+2$ then the last condition of~\eqref{system} only requires $L\le K+N-1$, and gives no other restrictions. Therefore we are left with $K+N$ unknown coefficients $c_{K+N-1},\ldots,c_0$. Note that again, this coincides with the number of the linear equations in~\eqref{system}.

Therefore in all cases the number of unknowns and the number of equations coincide. Moreover, the matrix of the coefficients is just the Vandermonde matrix with a nonzero determinant since all of $\{R_j\}, \{E_j\}$ are assumed to be pairwise different. Thus the solution is indeed unique. 

In fact, one can think about the system~\eqref{system} as a Lagrange interpolation problem or a Mittag--Leffler problem. Both have explicit solutions producing
\begin{equation}\label{m(z)}
m(z) = \frac{\sqrt{R_\fre(z)}-q(z)}{a(z)} + \sum_{j=1}^N \frac{c_j}{z-E_j} - \sum_{j=1}^M \frac{d_{j}}{z-R_j}  ,
\end{equation}
where
\begin{equation*}\label{c_j}
c_j = \frac{\sqrt{R_\fre((E_j)_+)}+q(E_j)}{a'(E_j)}, \quad  d_j = \frac{\sqrt{R_\fre((R_j)_+)}-q(R_j)}{a'(R_j)},
\end{equation*}
and
$$
q(z) =
\left\{
\begin{array}{ll}
0 & \mbox{if } K+N \ge l+2, \\
z^{l+1} & \mbox{if } K+N = l+1,\\
z^{l+1} - \tfrac12 z^l (\sum_{j=1}^{l+1} \alpha_j+\beta_j) & \mbox{if } K+N=l.\\
\end{array}
\right.
$$
Indeed, $c_j$'s are designed to make sure that there are no poles at $(E_j)_-$, and $d_{j}$'s that there are no poles at $(R_j)_+$. The $q(z)$ term is there to make sure that $m\to 0$ at $\infty_+$.


For the case when not all of $R_j$'s are simple resonances, the second line of~\eqref{system} needs to be changed to
$$
\tfrac{d^m}{dz^m} \left(p(z) +  \sqrt{R_\fre(z)} \right)\big|_{z=(R_j)_+} =0, \quad m=0,1,\ldots,n_j-1; j=1,\ldots,M,
$$
where $n_j$ is the multiplicity of $R_j$. In that case we still have the linear system of equations with the number of unknowns equal to the number of (non-identical) equations. The determinant of the matrix of coefficients can be shown to be equal to
$$
\pm \prod_{j<k} (R_j-R_k)^{n_j n_k} \prod_{j,k} (R_j-E_k)^{n_j} \prod_{j<k} (E_j-E_k),
$$
where in the products we do not repeat coinciding $R_j$'s. Indeed, this is just the so-called Hermite, rather than Lagrange, polynomial interpolation problem, which also has a unique solution. Finally, in~\eqref{m(z)} the last sum needs to be modified to $\sum_{j=1}^M \sum_{s=1}^{n_j} \frac{d_{j,s}}{(z-R_j)^s}$, where $d_{j,s}$ are the coefficients from the Laurent expansions
\begin{equation*}\label{d_j}
\frac{\sqrt{R_\fre(z)}-q(z)}{a(z)} = \sum_{s=1}^{n_j} \frac{d_{j,s}}{(z-R_j)^s} + O(1), \quad z\to (R_j)_+.
\end{equation*}


\bibliographystyle{plain}
\bibliography{../../mybib}

\end{document}